\documentclass[twocolumn,superscriptaddress,amsfont,amssymb,amsmath, showpacs,balancelastpage]{revtex4-1}
\usepackage{amsmath}
\usepackage{amsfonts}
\usepackage{amssymb}
\usepackage{bbm}
\usepackage[utf8]{inputenc}
\usepackage{braket}
\usepackage{environ}
\usepackage{graphicx,color,dsfont}
\usepackage{booktabs}
\usepackage{floatrow}
\usepackage[caption=false,label font={bf,normalsize}]{subfig}
\usepackage[dvipsnames]{xcolor}
\usepackage[all]{xy}
\usepackage{relsize}
\usepackage{hyperref}
\usepackage{amsthm}
\usepackage{accents}
\hypersetup{linktocpage}
\definecolor{darkblue}{RGB}{0,0,127} 
\definecolor{darkgreen}{RGB}{0,150,0}
\hypersetup{colorlinks=true,citecolor=red,linkcolor=darkblue, 
	filecolor=red, urlcolor=darkblue,pdftitle={Clifford Hierarchy Stabilizer Codes: Transversal Non-Clifford Gates and Magic},
pdfauthor={Ryohei Kobayashi, Guanyu Zhu, Po-Shen Hsin}}
\usepackage{lmodern}
\usepackage{mathtools}

\usepackage{amsthm}
\newtheorem{theorem}{Theorem}
\newtheorem{lemma}{Lemma}

\usepackage{lipsum}

\makeatletter
\def\@opargbegintheorem#1#2#3{\trivlist
   \item[]{\bfseries #1\ #2\ (#3)} \itshape}
\makeatother

\NewEnviron{eqs}{%
\begin{equation}\begin{split}
    \BODY
\end{split}\end{equation}
}

\newcommand{\Z}{\mathbb{Z}}
\def\Z{\mathbb{Z}}
\def\C{\mathcal{C}}

\newcommand{\beq}{\begin{equation}}
\newcommand{\eeq}{\end{equation}}
\newcommand{\bc}{\begin{center}}
\newcommand{\ec}{\end{center}}

\newcommand{\red}[1]{\textcolor{red}{#1}}
\newcommand{\blue}[1]{\textcolor{blue}{#1}}
\newcommand{\green}[1]{\textcolor{Green}{#1}}

\newcommand{\rd}{\textcolor{red}{r}}
\newcommand{\bl}{\textcolor{blue}{b}}
\newcommand{\gr}{\textcolor{Green}{g}}
\newcommand{\yl}{\textcolor{Dandelion}{y}}

\newcommand{\change}[1]{\textcolor{black}{#1}}

\begin{document}

\title{Clifford Hierarchy Stabilizer Codes: Transversal Non-Clifford Gates and Magic}

\author{Ryohei Kobayashi}
\email{ryok@ias.edu} 
\affiliation{School of Natural Sciences, Institute for Advanced Study, Princeton, NJ 08540, USA}

\author{Guanyu Zhu}
\email{guanyu.zhu@ibm.com}
\affiliation{IBM Quantum, IBM T.J. Watson Research Center, Yorktown Heights, NY 10598 USA}
\affiliation{Kavli Institute of Theoretical Physics, Santa Barbara, CA 93106, USA}

\author{Po-Shen Hsin}
\email{po-shen.hsin@kcl.ac.uk}

\affiliation{Department of Mathematics, King’s College London, Strand, London WC2R 2LS, UK}

\begin{abstract}
A fundamental problem in fault-tolerant quantum computation is the tradeoff between universality and dimensionality, exemplified by the the Bravyi-K\"onig bound for $n$-dimensional topological stabilizer codes.  In this work, we extend topological Pauli stabilizer codes to a broad class of $n$-dimensional Clifford hierarchy stabilizer codes. These codes correspond to the $(n+1)$D Dijkgraaf-Witten gauge theories with non-Abelian topological order.  We construct transversal non-Clifford gates through automorphism symmetries represented by cup products. In 2D, 
we obtain the first transversal non-Clifford logical gates including T and CS for Clifford stabilizer codes, using the automorphism of the twisted $\mathbb{Z}_2^3$ gauge theory (equivalent to $\mathbb{D}_4$ topological order). We also combine it with the just-in-time decoder to fault-tolerantly prepare the logical T magic state in $O(d)$ rounds via code switching.   In 3D, we construct a transversal logical $\sqrt{\text{T}}$ gate in a non-Clifford stabilizer code at the third level of the Clifford hierarchy, located on a tetrahedron corresponding to a  twisted  $\mathbb{Z}_2^4$ gauge theory.  
{Our constructions surpass the Bravyi-K\"onig bound by achieving the logical gates in the $(n+1)$-th level of Clifford hierarchy in $n$ spatial dimension.} 
\medskip
\noindent
\end{abstract}

 
 \date{\today}



\maketitle


\unitlength = .8mm

\setcounter{tocdepth}{3}

\bigskip

\textit{Introduction.}  The space-time overhead for achieving universality remain a central challenge in fault-tolerant quantum computation \cite{Kubica:2021}. The Bravyi-König theorem establishes that transversal logical gates--more generally, constant-depth circuits--in $n$-dimensional topological Pauli stabilizer codes are restricted in the $n^\text{th}$-level of Clifford hierarchy \cite{Bravyi:2013dx}.  For instance, realizing a non-Clifford gate at the third level necessitates at least a three-dimensional code, incurring $O(d^3)$ space or space-time overhead (with $d$ the code distance). An alternative approach using the ``just-in-time'' decoder \cite{bombin20182d, Browneaay4929} enables computation on a two-dimensional layout by emulating the 3D code within a (2+1)D spacetime framework, maintaining the same space-time overhead while reducing the spatial overhead to $O(d^2)$.  A fundamental question is general principles for lowering space or space-time overhead for logical non-Clifford gates. 

In recent years, symmetry has emerged as a powerful organizing principle for constructing transversal logical gates \cite{Yoshida_gate_SPT_2015, 
 Yoshida_global_symmetry_2016, Yoshida2017387, Zhu:2017tr, zhu:2022fractal, Barkeshli:2022wuz,Barkeshli:2022edm,zhu2023non,Barkeshli:2023bta,Kobayashi2024crosscap,Hsin2024_non-Abelian,Hsin2024:classifying, zhu2025topological, zhu2025transversal}. Higher-form symmetries provide addressable logical gates \cite{zhu2023non, Hsin2024:classifying}, while higher-group symmetries organize Clifford hierarchy of logical gates \cite{Barkeshli:2022wuz, Barkeshli:2022edm}. Non-invertible symmetry is important in gauging and measurements in lattice surgery \cite{Horsman_2012, williamson2024low,huang2025hybridlatticesurgerynonclifford} and code switching  \cite{Huang:2025cvt,Davydova:2025ylx,bauer2025planarfaulttolerantcircuitsnonclifford}. A large class of such gates are realized by constant-depth circuits and systematically classified in \cite{Hsin2024:classifying}, extending beyond the conventional color-code \cite{Kubica:2015mta}  or triorthogonality paradigm \cite{bravyi2012magic}.

In this work, we introduce a family of Clifford hierarchy stabilizer codes in $n$ spatial dimensions generalizing topological Pauli stabilizer codes, whose stabilizers lie in the $n^\text{th}$ level of Clifford hierarchy $\C_n$,  defined recursively on $m$-qubit \change{unitary} $U$, i.e., $\C_n:=\{U \in \mathbb{U}_m: \change{U \mathcal{P}_m U^\dagger} \subset \C_{n-1} \}$, where $\C_1=\mathcal{P}_m$ is the Pauli group and $\C_2$ is the Clifford group~\cite{Gottesman1999hierarchy}. 
\change{This generalizes XS, XP stabilizer codes using phase gates \cite{Ni2015XS, Webster2022XP}.}
We construct transversal logical T, CS gates in a 2D Clifford stabilizer code via the automorphism symmetry in the corresponding twisted $\mathbb{Z}_2^3$ gauge theory in (2+1)D. Such a code corresponds to a non-Abelian $\mathbb{D}_4$ topological order, realized on the ion-trap platform \cite{Iqbal:2023wvm}.  This work therefore represents a conceptual advance, providing the first demonstration of transversal non-Clifford gates in 2D Clifford stabilizer codes (see also the parallel work in \cite{Tyler2025groupsurfacecode}).

We then show how this transversal logical T gate can be deployed in a fault-tolerant protocol:
we combine the transversal T gate with code switching between a folded surface code and the non-Abelian code
via gauging measurements/condensation, together with the ``just-in-time'' decoder of Ref.~\cite{Davydova:2025ylx},
to prepare a logical $T$ magic state in $O(d)$ rounds.
Related works~\cite{bauer2025low,Davydova:2025ylx} interpret the spacetime protocols of
Refs.~\cite{bombin20182d,Browneaay4929} as a spacetime (2+1)D path integral in the same twisted $\mathbb{Z}_2^3$ gauge theory,
yielding a non-Clifford logical operation
(equivalently~\cite{magic_patent}).
\change{Our key conceptual advance is that we explicitly construct an invertible, finite-depth circuit
 acting entirely within the non-Abelian code---a concrete automorphism (invertible symmetry of a code)
that preserves the code space and realizes the transversal non-Clifford logical action directly,
rather than an effective spacetime implementation.}

This approach generalizes to a transversal logical $R_{N}=\text{diag}(1, e^{i 2\pi/2^N})$ gate in the $N^\text{th}$ level of Clifford hierarchy $\C_N$ in $N-1$ spatial dimensions, including the $\sqrt{\text{T}}$ gate at the fourth level $\mathcal{C}_4$ in 3D.
{Our construction surpasses the Bravyi-K\"onig bound for Pauli codes by one dimension.}

\change{
Interestingly, in spatial dimensions $n\ge 3$ we expect these non-Abelian codes to admit single-shot code switching with surface codes, driven by the confinement of magnetic flux loops (membranes) \cite{Bombin:2015hia}, analogous to 3D surface codes. More precisely, confinement suggests single-shot correction for non-Abelian fluxes, while a complete protocol must also correct the charge errors in a single shot. A natural route is a subsystem formulation in the spirit of gauge color codes, which can enable single-shot recovery via redundant local checks \cite{Bombin:2016dq}, or a dimensional-jump strategy \cite{Browneaay4929} adapted to the present dressed-stabilizer setting. Establishing such a fully single-shot switching/decoding procedure and its threshold is left for future work.
}

\textit{Clifford stabilizer code in 2D and T gate.}
Our stabilizer code in 2D is a (non-commuting) Clifford stabilizer code: 
\change{the generators $\{S_j\}$ of the stabilizer group $\mathcal{S}$ are local Clifford unitaries. The generators are not generally commutative, and the code space is spanned by the stabilizer states satisfying $S_j\ket{\psi}=\ket{\psi}$ for all $j$. This generalizes XS stabilizer codes \cite{Ni2015XS} to generic Clifford operators.}

We will present a Clifford stabilizer code that is a special case of those constructed in \cite{Hsin2024_non-Abelian} in generic dimensions, where we focus on 2D.
On the 2D lattice (triangulation), each edge has 3 qubits, whose Pauli $Z$ eigenvalues label the gauge fields $Z^{\red{r}}=(-1)^{a_{\red{r}}},Z^{\green{g}}=(-1)^{a_{\green{g}}},Z^{\blue{b}}=(-1)^{a_{\blue{b}}}$ with $a_{\red{r}},a_{\green{g}},a_{\blue{b}}$ being operator-valued 1-cochains with eigenvalues $0,1$.
Similarly, there are Pauli $X$ operators $X^{\red{r}},X^{\green{g}},X^{\blue{b}}$ on each edge. The Clifford stabilizer is generated by
\begin{align}
    \mathcal{S} = \left\{ \mathcal{S}_{X}^{\red{r}},\mathcal{S}_{X}^{\green{g}},\mathcal{S}_{X}^{\blue{b}},\mathcal{S}_{Z}^{\red{r}},\mathcal{S}_{Z}^{\green{g}},\mathcal{S}_{Z}^{\blue{b}}\right\}~,
\end{align}
with the stabilizer $\mathcal{S}_{X}$ supported at vertices $v$ of the 2D lattice. $X$-stabilizers are generated by $S^{\red{r}}_{X;v},S^{\green{g}}_{X;v},S^{\blue{b}}_{X;v}$ on each vertex $v$ with the form
\change{
\begin{align}
\begin{split}
    S_{X;v}^{\red{r}} &= \left(\prod_{\partial e\supset v} X^{\rd}_e\right) \prod_{e',e'':\int\tilde v\cup \tilde e'\cup \tilde e''\neq 0}CZ^{\gr,\bl}_{e',e''}~, \\
    S_{X;v}^{\gr} &= \left(\prod_{\partial e\supset v} X^{\gr}_e\right) \prod_{e',e'': \int\tilde e''\cup \tilde v\cup \tilde e'\neq 0}
   CZ^{\bl,\rd}_{e,e'}~,\\
   S_{X;v}^{\bl} &= \left(\prod_{\partial e\supset v} X^{\bl}_e\right)\prod_{e',e'':\int \tilde e'\cup \tilde e''\cup \tilde v\neq 0} CZ^{\rd,\gr}_{e',e''}~,
   \label{eq:Sr v}
   \end{split}
\end{align}
}
where $\cup$ denotes cup product of cochains (see SM Sec.~\ref{app:cup} for a review). 
\change{$\tilde v$ is a 0-cochain that takes value 1 on vertex $v$ and 0 on other vertices, and $\tilde e$ is a 1-cochain that takes value $1$ on edge $e$ and 0 on other edges. $CZ^{c,c'}_{e,e'}$ is the CZ gate for a pair of qubits with colors $c,c'$, on edges $e,e'$.}
$\mathcal{S}_Z^{\red{r}}$ is generated by $S^{\red{r}}_{Z;f}$ on each face with
\begin{align}
    S^{\red{r}}_{Z;f} = \prod_{e\in \partial f} Z_e^{\red{r}}~.
\end{align}
Other $Z$-stabilizers with colors $\green{g},\blue{b}$ are defined in the same fashion. The $X$-stabilizers with different colors (e.g., $\mathcal{S}_{X}^{\red{r}},\mathcal{S}_{X}^{\green{g}}$) are commutative in the $Z$-stabilizer subspace of $\mathcal{S}_{Z}^{\red{r}}, \mathcal{S}_{Z}^{\green{g}},\mathcal{S}_{Z}^{\blue{b}}$, therefore $\mathcal{S}$ forms a non-commuting stabilizer group.
{See Fig.~\ref{fig:trianglehamiltonian} for illustrations of $X,Z$-stabilizers on a square lattice.}

The code space is equivalent to a twisted $\mathbb{Z}_2^3 = \Z_2^{\red{r}}\times\Z_2^{\green{g}}\times\Z_2^{\blue{b}}$ gauge theory, with the Dijkgraaf-Witten twist $(-1)^{\int a_{\rd}\cup a_{\gr}\cup a_{\bl}}$.
This theory has 22 anyons including $e_{\rd}, e_{\gr}, e_{\bl}, m_{\rd}, m_{\gr}, m_{\bl}$, whose notations are aligned with \cite{Iqbal:2023wvm}; $m$ denotes magnetic fluxes, and $e$ denotes electric charges of $\Z_2^3$ gauge theory. $e$, $m$ with the same color has $(-1)$ mutual braiding due to the Aharanov-Bohm phase.

The underlying gauge group $\mathbb{Z}_2^3$ has automorphism permuting group elements. 
We focus on the automorphism transforming the elements as 
\begin{equation}\label{eqn:automorphism}
    (g_{\red{r}},g_{\green{g}},g_{\blue{b}})\in \{0,1\}^3\rightarrow (g_{\red{r}}, g_{\rd} + g_{\gr} + g_{\bl} , g_{\bl})~,
\end{equation}
which transforms the gauge groups as $\Z_2^{\rd}\leftrightarrow \Z_2^{\rd\gr}:= \text{diag}(\Z_2^{\rd},\Z_2^{\gr})$, $\Z_2^{\gr}\leftrightarrow \Z_2^{\gr}$,  $\Z_2^{\bl}\leftrightarrow\Z_2^{\gr\bl}$.

We show that the above automorphism induces an emergent symmetry on the Clifford stabilizers, i.e. preserves the logical subspace.
The automorphism symmetry $U$ is constructed from the product of transversal automorphism $V$ and another operator $W$, where
\begin{align}
    V= \bigotimes_e \text{CNOT}^{(\rd,\gr)}_{e} \text{CNOT}^{(\bl,\gr)}_{e}~,
    \label{eq:transversalCNOT}
\end{align}
which transforms the Pauli $X$ operators according to the automorphism (\ref{eqn:automorphism}): $X_e^{\rd} \leftrightarrow X_e^{\rd} X_e^{\gr}$ and $X_e^{\bl} \leftrightarrow X_e^{\bl} X_e^{\gr}$. Pauli $Z$ operators are transformed as $Z_e^{\gr} \leftrightarrow Z_e^{\rd} Z_e^{\gr} Z_e^{\bl}$.
The operator $V$ induces the desired automorphism \eqref{eqn:automorphism} of the gauge group. 
\change{However, $V$ is not yet an emergent symmetry of the Clifford stabilizers; for instance, while Pauli $X$ operators are transformed as $X_e^{\rd} \leftrightarrow X_e^{\rd} X_e^{\gr}$, the $X$-stabilizers are not preserved by $V$: $VS^{\rd}_{X;v}V^\dagger \neq S^{\rd}_{X;v}S^{\gr}_{X;v}$.
The true symmetry hence requires additional modification $W$ to be specified below.}

The true emergent symmetry $U$ is obtained by dressing $V$ with additional transversal operators. While the symmetry operator $U$ can be defined on generic triangulations,
on a 2D square lattice it is
\begin{equation}
\begin{split}
\label{eqn:automorphismsymmetryoperator}
     U =  WV~, \quad 
     W & = \exp\left(\pi i \int \frac{a_{\rd} \cup a_{\bl}}{2}\right) \\
       &= \prod_{p=(0123)} \text{CS}^{\rd,\bl}_{e_{01},e_{13}}(\text{CS}^{\rd,\bl}_{e_{02},e_{23}})^\dagger~,
     \end{split}
\end{equation}
where $(a_{\rd} \cup a_{\bl})/2$ is a 2-cochain integrated over a 2d space. $\text{CS}^{\rd,\bl}_{e,e'}$ is a controlled-S gate for the red qubit at the edge $e$ and blue at $e'$, multiplied over plaquettes with vertices labeled by $0,1,2,3$ (see Fig.~\ref{fig:CS}). The operator $U$ then preserves the stabilizer group, e.g., $US^{\rd}_{X;v}U^\dagger = S^{\rd}_{X;v}S^{\gr}_{X;v}$, see SM Sec.~\ref{sec:proof of theorem 1} for derivations.
This extra operator $W$ plays a crucial role for getting non-Clifford logical action.
As shown in SM Sec.~\ref{app:Tgate}, the unitary operator $U$ preserves the logical subspace:
\begin{theorem}
    The unitary operator $U$ preserves the logical subspace and thus is a logical gate of the Clifford stabilizer code.
\end{theorem}
In the following, we describe a setup where $U$ implements a logical T gate.

\begin{figure}[t]
    \centering
    \includegraphics[width=0.5\linewidth]{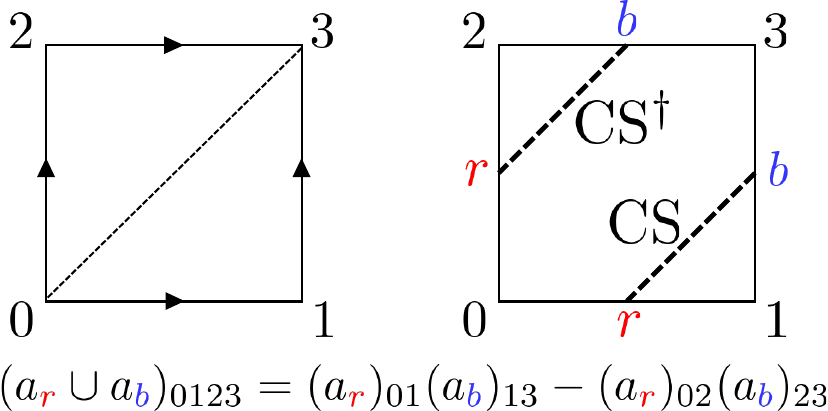}
    \caption{The operator $W$ is expressed using cup product of cochains $a_{\rd}, a_{\bl}$ for $\Z_2$ gauge fields. This operator corresponds to a product of CS, CS$^\dagger$ operators on each plaquette.  }
    \label{fig:CS}
\end{figure}

\medskip

Consider the Clifford stabilizer on an triangle with three gapped boundaries. 
Gapped boundaries are characterized by
symmetry breaking from the bulk $\mathbb{Z}_2^3$ gauge group to subgroups on the boundaries \cite{Beigi_2011,Hsin2024:classifying}; the first boundary ${\cal L}_{\rd}$ breaks $\mathbb{Z}_2^3$ to $\mathbb{Z}_2^{\gr}\times\mathbb{Z}_2^{\bl}$. The second boundary ${\cal L}_{\bl}$ breaks the gauge group to $\mathbb{Z}_2^{\rd}\times \Z_2^{\gr}$, while the third boundary ${\cal L}_{\rd\bl}$ breaks to the diagonal $\mathbb{Z}^{\rd\bl}_2$.
These symmetry breaking at boundaries are represented by the boundary conditions of gauge fields,
\begin{align}
    \begin{split}
        \mathcal{L}_{\rd}: a_{\rd} = 0~, \quad \ \mathcal{L}_{\bl}: a_{\bl} = 0~, \quad 
        \mathcal{L}_{\rd\bl}: a_{\rd} + a_{\bl} = a_{\gr} = 0~. \\
    \end{split}
    \label{eq:2Dgaugefields boundaryconditions}
\end{align}
In the terminology of \cite{Kitaev_2012} for rough and smooth boundaries, the boundary ${\cal L}_{\rd}$ is the rough boundary for $\mathbb{Z}^{\rd}_2$ and smooth boundary for $\mathbb{Z}_2^{\gr}\times\mathbb{Z}_2^{\bl}$; the boundary ${\cal L}_{\bl}$ is rough for $\mathbb{Z}_2^{\bl}$ and smooth for $\mathbb{Z}_2^{\rd}\times \Z_2^{\gr}$; ${\cal L}_{\rd\bl}$ is rough for $\mathbb{Z}_2^{\gr}$ and a new ``mixed'' boundary for $\mathbb{Z}_2^{\rd},\mathbb{Z}_2^{\bl}$ where the combined electric charge $e_{\rd\bl}$ and combined magnetic flux $m_{\rd\bl}$ are condensed, as well as the excitations formed by their products. The boundary stabilizers on a square lattice with this triangular boundaries are shown in Fig.~\ref{fig:trianglehamiltonian}.

\begin{figure}[htb]
    \centering
    \includegraphics[width=0.8\linewidth]{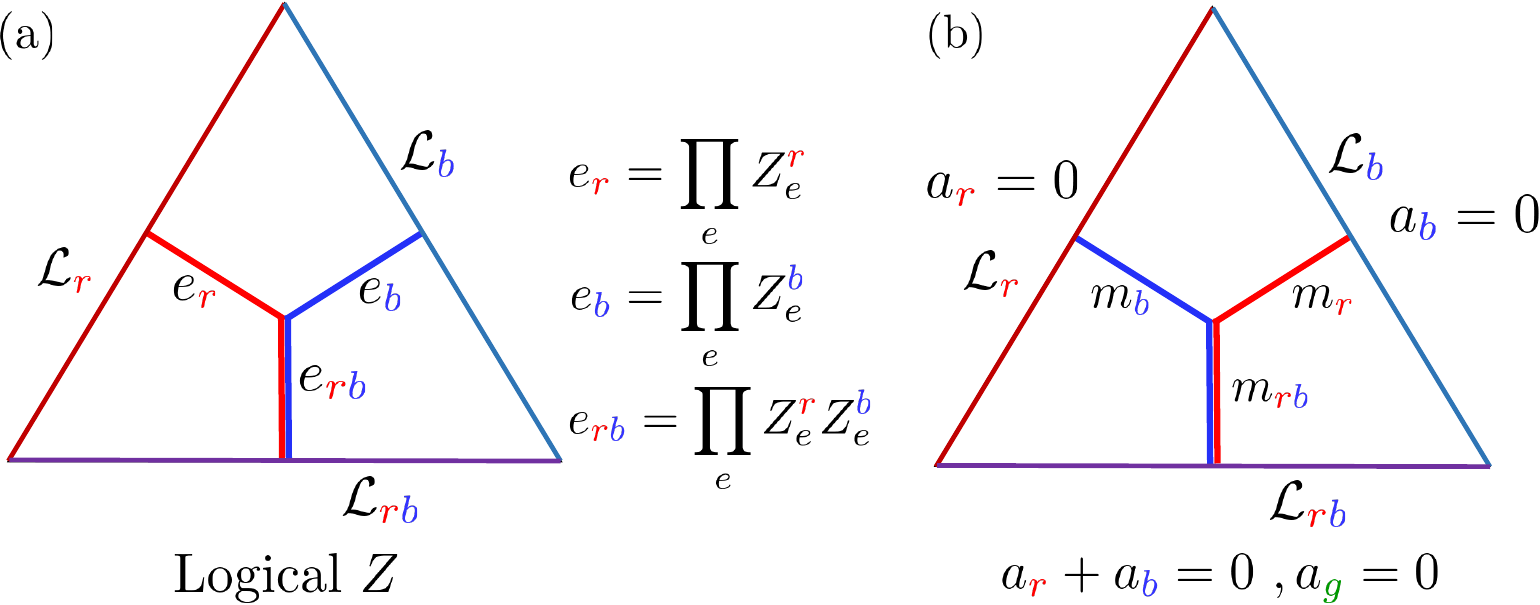}
    \caption{The stabilizer model is defined on a triangle with gapped boundaries $\mathcal{L}_{\rd}, \mathcal{L}_{\bl}, \mathcal{L}_{\rd\bl}$. (a): The logical $Z$ operator is a junction of string operators $e_{\rd}, e_{\bl}, e_{\rd\bl}$ ending at boundaries. Each electric charge $e_{\rd}, e_{\bl}, e_{\rd\bl}$ corresponds to a product of Pauli $Z$ operators along the string. (b): There is the other topological operator formed by a junction of $m_{\bl}, m_{\rd}, m_{\rd\bl}$. This anti-commutes with the logical Pauli $Z$ operator. The figure also shows the boundary conditions of gauge fields. }
    \label{fig:triangle}
\end{figure}

The model on this triangle region has a Hilbert space of a single logical qubit: 
\begin{lemma}
The Clifford stabilizer model with the boundary condition has a single logical qubit.
\end{lemma}
\begin{proof}
    The state is generally labeled by the eigenvalues of nontrivial electric charge operators $e_{\rd},e_{\gr},e_{\bl}$, which characterizes the configuration of $\Z_2^3$ gauge fields on a space.
    On the triangle, the only non-contractible Pauli $Z$ string operator is the junction of electric charges $e_{\rd}, e_{\bl}, e_{\rd\bl}$ shown in Fig.~\ref{fig:triangle} (a). To see that this Pauli $Z$ operator can have nontrivial eigenvalues, let us consider the other topological operator formed by the junction of $m_{\bl}, m_{\rd}, m_{\rd\bl}$ as shown in ~\ref{fig:triangle} (b). This anti-commutes with the Pauli $Z$ operator, so the Hilbert space has a single qubit $\{\ket{0},\ket{1}\}$ labeled by the Pauli $Z$ eigenvalue.

\end{proof}

The automorphism symmetry $U$ preserves the boundaries: the unbroken gauge group for each of the boundaries at each boundary, $\Z_2^{\gr}\times \Z_2^{\bl}, \Z_2^{\rd}\times\Z_2^{\gr}, \Z_2^{\rd\bl}$, is invariant under the automorphism \eqref{eqn:automorphism}. 
Thus the automorphism symmetry $U$ 
generates a logical gate with boundaries.

\begin{theorem}
    The automorphism symmetry $U$ in the Clifford stabilizer model with the boundary condition implements logical T gate.
\end{theorem}
The detailed proof is in SM Sec.~\ref{app:Tgateboundary}. To obtain the non-Clifford action, it is essential that in the presence of gapped boundaries, the logical operator $U=WV$ acquires boundary modifications \cite{Hsin2024:classifying}. For instance, on a square lattice with boundary conditions shown in Fig.~\ref{fig:trianglehamiltonian}, the logical operator is $U=WV$ with $V$ given by \eqref{eq:transversalCNOT}, and $W$ gets modified as
\begin{equation}\label{eqn:automorphismsymmetryboundary}
     W = 
     \prod_{p=(0123)} \text{CS}^{\rd,\bl}_{e_{01},e_{13}}(\text{CS}^{\rd,\bl}_{e_{02},e_{23}})^\dagger \times \prod_{e\in \text{bdry}_{\rd\bl}} (\text{T}^{\rd}_{e})^\dagger~,
\end{equation}
where the second product is over edges on the boundary $\mathcal{L}_{\rd\bl}$. As shown in SM Sec.~\ref{app:Tgateboundary}, these boundary operators contribute as a logical T$^\dagger$ gate;
the automorphism symmetry $U$ acts on the logical qubit in the Pauli $Z$ basis as $U\overline{|m\rangle}= e^{-\frac{\pi i m}{4}}\overline{|m\rangle}=\overline{\text{T}}^\dagger\overline{|m\rangle}$, $m=0,1$.

Let us compare the transversal logical T gate with the construction in \cite{Davydova:2025ylx}, where a T gate of a folded surface code on a triangle is obtained via code switching through the $\mathbb{D}_4$ gauge theory. The code switching corresponds to the action of the gapped domain wall separating the $\Z_2^2$ gauge theory for a folded surface code and the $\mathbb{D}_4$ gauge theory; in the spacetime picture, $\mathbb{D}_4$ gauge theory is sandwiched by a pair of gapped domain walls DW and DW', see Fig.~\ref{fig:switching}. In fact, the two domain walls DW, DW' are related by the transversal unitary operator $\tilde U$ with $\text{DW'} = \text{DW} \times \tilde U$,
where in their protocol, $\tilde U$ acts by an automorphism acting on the gauge groups as $\Z_2^{\rd}\leftrightarrow \Z_2^{\rd}, \Z_2^{\gr}\leftrightarrow \Z_2^{\bl}$.
Also in their protocol, the $\mathbb{D}_4$ gauge theory on a triangle is bounded by the following three boundary conditions $\tilde{\mathcal{L}}_{\rd}, \tilde{\mathcal{L}}_{\bl},\tilde{\mathcal{L}}_{\rd\bl}$ with the unbroken gauge groups
\begin{align}
    \tilde{\mathcal{L}}_{\rd}: \Z_2^{\gr}\times \Z_2^{\bl}, \quad  \tilde{\mathcal{L}}_{\bl}: \Z_2^{\rd\gr}\times \Z_2^{\rd\bl}, \quad \tilde{\mathcal{L}}_{\rd\bl}: \Z_2^{\rd}.
\end{align}

One can then see that the above unitary $\tilde U$, together with the gapped boundaries $\tilde{\mathcal{L}}_{\rd},\tilde{\mathcal{L}}_{\bl},\tilde{\mathcal{L}}_{\rd\bl}$ are transformed into the ones without tilde by an automorphism of gauge groups
\begin{align}
    (g_{\red{r}},g_{\green{g}},g_{\blue{b}})\in \{0,1\}^3\rightarrow (g_{\red{r}},  g_{\gr}, g_{\rd} + g_{\gr} + g_{\bl})~,
    \label{eq:autofortilde gauge fields}
\end{align}
which transforms the gauge groups as $\Z_2^{\rd}\leftrightarrow \Z_2^{\rd\bl}$, $\Z_2^{\gr}\leftrightarrow \Z_2^{\gr\bl}$,  $\Z_2^{\bl}\leftrightarrow\Z_2^{\bl}$, inducing an automorphism of $\mathbb{D}_4$ gauge theory.
Therefore, by interpreting the domain walls $\text{DW'}\times (\text{DW})^\dagger$ as the action of transversal unitary $\text{DW}\times \tilde U\times  (\text{DW})^\dagger$, the action of $\tilde U$ on the $\mathbb{D}_4$ gauge theory is identified as a logical T gate $U$ by an automorphism \eqref{eq:autofortilde gauge fields}. 

\begin{figure}[htb]
    \centering
    \includegraphics[width=0.6\linewidth]{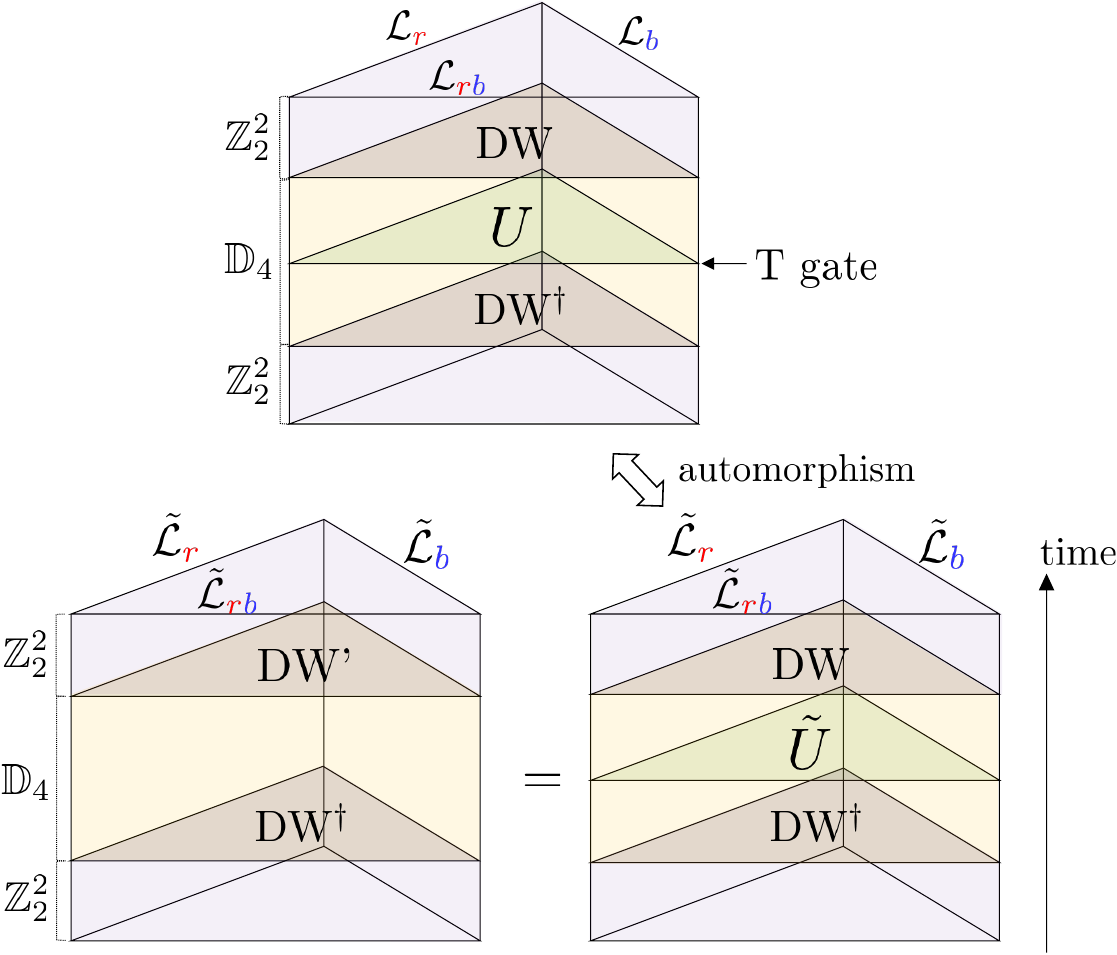}
    \caption{Bottom left: The code switching is understood as a combination of gapped domain walls $\text{DW'}\times \text{DW}^\dagger$, where each domain wall separates $\mathbb{D}_4$ and $\Z_2^2$ gauge theory. Bottom right: The domain wall DW' is equivalent to a combination $\text{DW}\times \tilde U$ with an transversal unitary $\tilde U$. Top: By an automorphism \eqref{eq:autofortilde gauge fields}, the operator $\tilde U$ is identified as the logical T gate $U$.}
    \label{fig:switching}
\end{figure}

Within our stabilizer code, the $\text{T}$ magic state can be fault-tolerantly prepared by the following code switching along with the logical T gate:
\begin{enumerate}
    \item We start with a logical state $\overline{\ket{+}}$ of the  $\Z_2^{\rd}\times\Z_2^{\bl}$ gauge theory with boundaries, whose stabilizers are simply obtained by eliminating green qubits from Fig.~\ref{fig:trianglehamiltonian}. 
    The code space is equivalent to a folded surface code.
    We also initialize the green qubits as $\bigotimes_e \ket{0}^{\gr}_e$. 
    \item We then perform $O(d)$ rounds of syndrome measurements ($d$ is the distance) of the unknown Clifford stabilizers $S_{X;v}^{\gr}$ (with $\pm 1$ random eigenvalues) and the other Clifford and Pauli $Z$ stabilizers in the entire code; \change{this prepares a state of the Clifford stabilizer code \cite{Tantivasadakarn2024LRE, verresen2022efficient, Tantivasadakarn2023shortest}.}
    Meanwhile, we apply the ``just-in-time" decoder from Ref.~\cite{Davydova:2025ylx}.  The non-abelian flux errors \change{violating the plaquette $Z$-stabilizers} 
    are paired up  in real time  to form closed flux worldline whenever the observation time of the syndromes is comparable to the separation of the syndromes. After the $O(d)$ rounds of syndrome measurements, we use the matching or RG decoder \cite{ducloscianci2010renormalizationgroupdecodingalgorithm} to clean up the remaining Abelian charge syndromes and the corresponding $Z$ errors.  The above procedure is also called a gauging procedure, where one effectively switches the code space into the $\mathbb{D}_4$ code.
    \item We act the logical T$^\dagger$ gate $U$.
    \item We then measure the Pauli $Z^{\gr}_e$ operators on each edges, followed by acting $S_X^{\gr}$ to flip the strings of $\ket{1}^{\gr}_e$ states. This switches the code back to the $\Z_2^{\rd}\times\Z_2^{\bl}$ surface code followed by $O(d)$ rounds of error corrections.
    The final state is given by $\overline{\text{T}}^\dagger\overline{\ket{+}}$.
\end{enumerate}

In addition to T gate, we construct transversal logical CS gate using six layers of the 2D Clifford stabilizer model.
We note that the CS gate in the 3D color code has been constructed in Ref.~\cite{brown2025colorcodelogicalcontrols}.
See SM Sec.~\ref{sec:CS} for the construction of the logical CS gate in 2D. 

\begin{figure}[htb]
    \centering
    \includegraphics[width=0.75\linewidth]{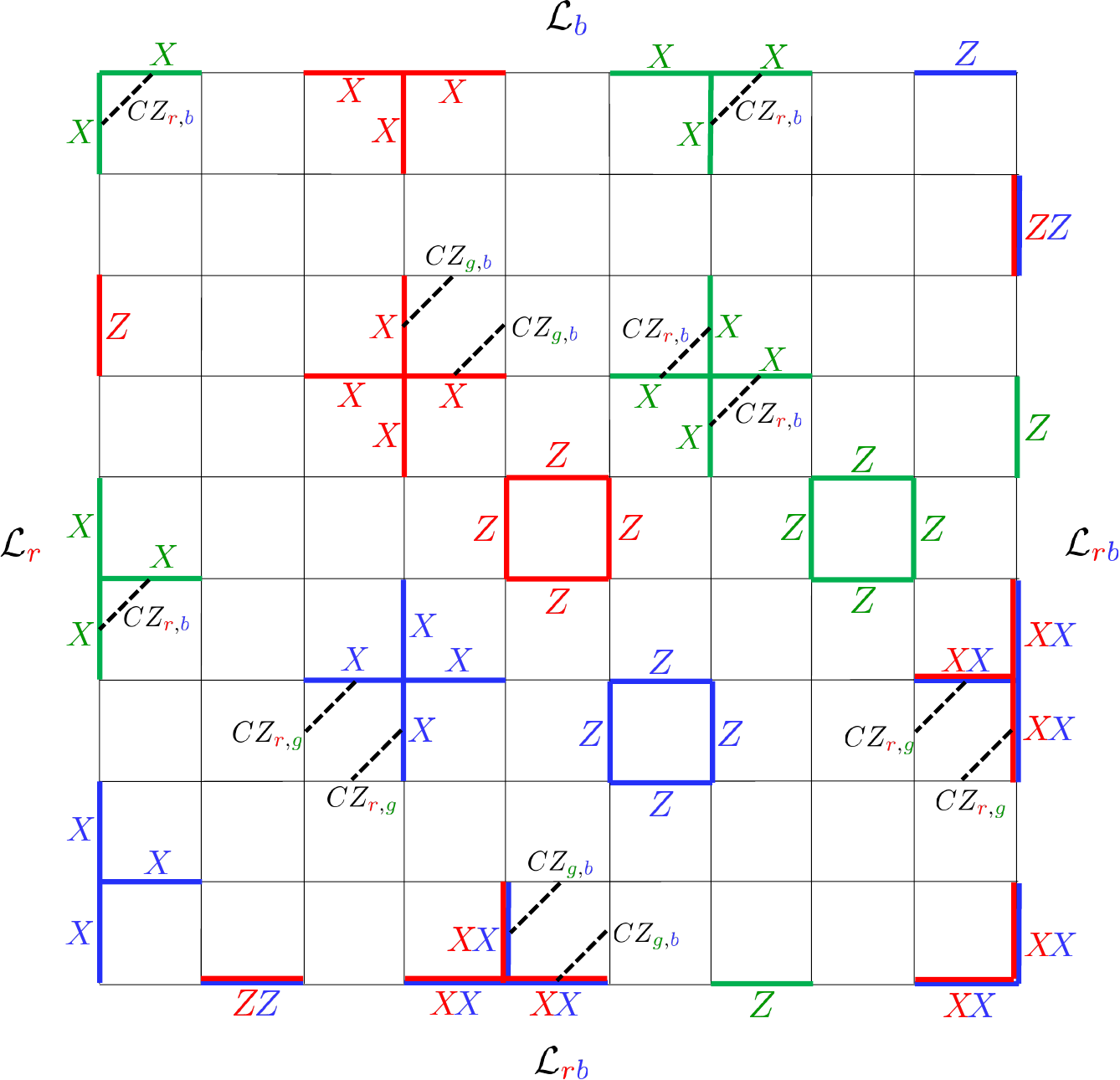}
    \caption{The stabilizer on a square lattice surrounded by three gapped boundaries. The bottom boundary of the triangle $\mathcal{L}_{\rd\bl}$ is realized at the bottom and right boundaries of the square, therefore a square is bounded by three gapped boundaries. 
    }
    \label{fig:trianglehamiltonian}
\end{figure}

\textit{Non-Clifford stabilizer code in 3D and $\sqrt{T}$ gate.} 
We extend our construction to 3D, where we construct a non-Clifford $\sqrt{\text{T}}$ gate at 4th level of the Clifford hierarchy in a $\Z_2^4$ twisted gauge theory. The Dijkgraaf-Witten twist is given by $(-1)^{\int a_{\rd}\cup a_{\yl}\cup a_{\bl}\cup a_{\gr} }$.
On a 3D lattice, we now have four physical qubits colored by $\rd,\gr,\bl,\yl$ on each edge.
The non-Clifford stabilizer is generated by
\begin{align}
    \mathcal{S} = \left\{ \mathcal{S}_{X}^{\red{r}},\mathcal{S}_{X}^{\green{g}},\mathcal{S}_{X}^{\blue{b}},\mathcal{S}_{X}^{\yl}, 
    \mathcal{S}_{Z}^{\red{r}},\mathcal{S}_{Z}^{\green{g}},\mathcal{S}_{Z}^{\blue{b}},\mathcal{S}_{Z}^{\yl}\right\}~,
\end{align}
where the stabilizers $\mathcal{S}_X$ are supported at vertices. For instance, $\mathcal{S}_{X}^{\red{r}}$ is generated by $\mathcal{S}_{X}^{\red{r};v}$ on each vertex as
\change{
\begin{align}
      S_{X;v}^{\red{r}} = \left(\prod_{\partial e\supset v} X^{\rd}_e\right) \prod_{\substack{e_1,e_2,e_3: \\ \int\tilde v\cup \tilde e_1\cup \tilde e_2\cup \tilde e_3\neq 0}}CCZ^{\yl,\bl,\gr}_{e_1,e_2, e_3}~,
\end{align}
}
defined in a similar fashion as the 2D case.
$\mathcal{S}_Z$ again are generated by the stabilizers on faces in the form of e.g., 
\begin{align}
    S^{\red{r}}_{Z;f} = \prod_{e\in \partial f} Z_e^{\red{r}}~.
\end{align}
See SM Sec.~\ref{app:sqT} for a complete description of stabilizers.

To construct a $\sqrt{\text{T}}$ gate, we locate a code on a 3D tetrahedron with four gapped boundaries on faces.
Each gapped boundary condition is specified by the boundary conditions of gauge fields:
\begin{align}
\nonumber    \begin{split}
        \text{1st, 2nd, and 3rd boundary}\quad & a_{\bl} = 0,  \  a_{\rd} = 0, \ a_{\gr} = 0, \\
        \text{4th boundary}\quad & a_{\rd}+a_{\gr} + a_{\bl} = a_{\yl} = 0~, \\
    \end{split}
\end{align}
which naturally generalizes \eqref{eq:2Dgaugefields boundaryconditions}.
These boundary conditions are enforced by boundary stabilizers
$Z_{\bl}$ (1st boundary), $Z_{\rd}$ (2nd boundary), $Z_{\gr}$ (3rd boundary), and $Z_{\rd}Z_{\gr}Z_{\bl}, Z_{\yl}$ (4th boundary) on boundary edges. The $X$ stabilizers at the boundaries are obtained by truncating the stabilizers in $\mathcal{S}_X$ that commute with the above boundary $Z$ stabilizers. This naturally generalizes the construction of boundary $X$ stabilizers shown in Fig.~\ref{fig:trianglehamiltonian} to 3D.
\change{The stabilizer code with above boundaries stores a single logical qubit; it has a single nontrivial $\Z_2^4$ gauge field configuration on a tetrahedron that corresponds to the logical $\ket{\overline{1}}$ state.} 
Similar to the 2D case, there is a single logical $Z$ operator formed by a junction of $Z^{\rd}, Z^{\gr}, Z^{\bl}$ strings. See SM Sec.~\ref{app:sqT} for detailed discussions.

The $\sqrt{\text{T}}$ gate again has an expression $U=WV$,
with
\begin{align}
   V= \bigotimes_e \text{CNOT}^{(\rd,\yl)}_{e} \text{CNOT}^{(\gr,\yl)}_{e}\text{CNOT}^{(\bl,\yl)}_{e}~,
\end{align}
\change{and $W$ is expressed by a product of CCS gates in the 3D bulk, CT$^\dagger$ gates on the 4th boundary, and $\sqrt{\text{T}}$ gates on the 1D hinge between the 1st, 4th boundary. In SM Sec.~\ref{app:sqT}, we present an explicit form of $W$ and demonstrate that $U$ generates a logical $\sqrt{\text{T}}$ gate.}

In conclusion, we discover transversal non-Clifford logical gates in Clifford hierarchy stabilizer codes, including  T, CS gates in 2D and $\sqrt{\text{T}}$ gate in 3D. 
The code can be further generalized to twisted $\mathbb{Z}_2^N$ gauge theory in $(N-1)$ spatial dimensions with the Dijkgraaf-Witten twist $(-1)^{a_1\cup a_2\cup\cdots a_N}$ for the $\mathbb{Z}_2^N$ gauge fields $\{a_i\}_{i=1}^N$. In such cases, the $\mathcal{S}_X$ stabilizers have gates in the $(N-1)^\text{th}$ level of Clifford hierarchy, while $Z$-stabilizers are Pauli operators. Its automorphism symmetry can realize $\overline{R_{N}}$ logical gate. 
{Our construction for the family of $\overline{R_{N}}$ logical gates surpasses the  Bravyi-K\"onig bound of Pauli codes.}
Higher-hierarchy codes require increased local resources as $N$ increases; the generalized Bravyi-K\"onig bound can be viewed as a tradeoff between spatial dimension and local complexity in non-Pauli stabilizers.

\section*{Acknowledgment}

We thank Andreas Bauer, Tyler Ellison, Sheng-Jie Huang, Sakura Schafer-Nameki, Nat Tantivasadakarn, Julio Magdalena de la Fuente, Zhenghan Wang, Dominic Williamson for discussions. We thank Julio Magdalena de la Fuente for comments on the draft.
R.K. is supported by the U.S. Department of Energy through grant number DE-SC0009988 and the Sivian Fund.
G.Z. is supported by the U.S. Department of Energy, Office of Science, National Quantum
Information Science Research Centers, Co-design Center
for Quantum Advantage (C2QA) under contract number
DE-SC0012704. G.Z was also  supported in part by grant NSF PHY-2309135 to the Kavli Institute for Theoretical Physics (KITP), where the work was completed.
P.-S.H. is supported by Department of Mathematics King’s College
London. 
R.K. and P.-S.H. thank Isaac Newton Institute for hosting the workshop ``Diving Deeper into Defects: On the Intersection of Field Theory, Quantum Matter, and Mathematics'', during which part of the work is completed. Near the completion of our work, we learned of another upcoming work \cite{Tyler2025groupsurfacecode} that also discusses transversal non-Clifford gates in related codes in 2D.

\bibliographystyle{utphys}
\bibliography{biblio, mybib_merge}

\onecolumngrid

\vspace{0.3cm}

\newpage

\begin{center}
\Large{\bf Supplemental Materials}
\end{center}
\onecolumngrid

\section{Review: Cup Product on Triangulations and Lattices}
\label{app:cup}

The cup product (see e.g. \cite{milnor1974characteristic}) provides a way to combine $\mathbb{Z}_N$-valued cochains on a triangulated manifold into higher-degree cochains. An $m$-cochain $f_m$ is a function assigning a $\mathbb{Z}_N$ value to each oriented $m$-simplex (or $m$-cell). Given two such cochains, $f_m$ and $g_n$, their cup product $f_m \cup g_n$ defines an $(m+n)$-cochain. 
To define a cup product on a triangulated manifold, it is necessary to specify a branching on the triangulation. A branching structure is an assignment of an orientation to each edge of every simplex such that no oriented loop exists within any 2-simplex. Then, each vertex of a single $k$-simplex is labeled by $(0,1,2,\dots,k)$ according to the induced ordering of vertices.

On a triangulated manifold with a branching structure, where each simplex is labeled by ordered vertices $(0,1,2,\dots,k)$, the cup product is evaluated locally on an $(m+n)$-simplex as  
\begin{equation}
(f_m \cup g_n)(0,1,2,\dots,m+n)
= f_m(0,1,\dots,m)\, g_n(m,m+1,\dots,m+n)~.
\end{equation}

For hypercubic lattices, where the basic cells are unit hypercubes $s_k$ given by $[0,1]^k$ with coordinates $(x^1, x^2, \dots, x^k)$, the construction is simply obtained by triangulating a hypercube into simplices~\cite{chen2021higher}. On an $(m+n)$-dimensional hypercube $s_{m+n}$, the cup product is defined by  
\begin{equation}
(f_m \cup g_n)(s_{m+n})
= \sum_I f_m([0,1]^I)\, g_n((x^I=1, x^{\bar I}=0) + [0,1]^{\bar I})~,
\end{equation}
where the sum runs over all subsets $I$ of $\{1,2,\dots,m+n\}$ with $|I|=m$, and $\bar I$ denotes the complement of $I$. Each term corresponds to a pair of lower-dimensional hypercubes: $f_m$ is evaluated on an $m$-dimensional face $[0,1]^I$ that begins at $(x^I=0, x^{\bar I}=0)$, while $g_n$ is evaluated on an $n$-dimensional face $(x^I=1, x^{\bar I}=0) + [0,1]^{\bar I}$ starting from the corner where $x^I=1$ and $x^{\bar I}=0$. 

\subsection{Integral of cochains}

Let us consider a $k$-dimensional triangulated manifold $M_k$ with a branching structure. A branching structure induces an orientation $\sigma(\Delta_k)=\pm 1$ on each $k$-simplex $\Delta_k$, see Fig.~\ref{fig:orientation} for $k=2$.
One can then integrate a $\Z_N$ $k$-cochain over a $k$-manifold,
\begin{align}
    \int_{M_k} f_k := \sum_{\Delta_k} \sigma(\Delta_k) f_k~,
\end{align}
which is valued in $\Z_N$.

\begin{figure}[htb]
    \centering
    \includegraphics[width=0.3\linewidth]{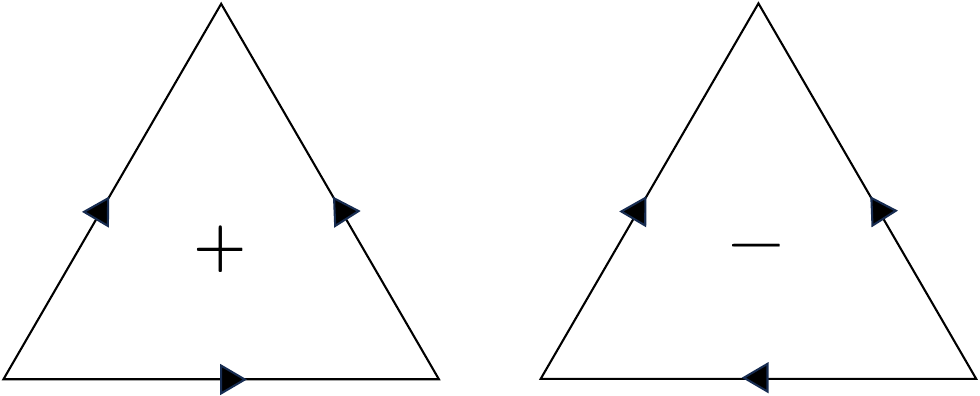}
    \caption{A branching structure introduces an orientation ($+,-$) of 2-simplices.
    }
    \label{fig:orientation}
\end{figure}

\section{Detail of Clifford Stabilizer model in 2D}
\label{app:2Dclifford}

The Clifford stabilizer model is a special case for the models constructed in \cite{Hsin2024_non-Abelian}.
On the 2D lattice, each edge has 3 qubits, whose Pauli $Z$ eigenvalues label the gauge fields $Z^{\rd}=(-1)^{a_{\rd}},Z^{\gr}=(-1)^{a_{\gr}},Z^{\bl}=(-1)^{a_{\bl}}$ with $a_{\rd},a_{\gr},a_{\bl}$ being 1-cochains with eigenvalues $0,1$.
Similarly, there are Pauli $X$ operators $X^{\rd},X^{\gr},X^{\bl}$ on each edge. The Clifford stabilizer is given by
\begin{align}
    \mathcal{S} = \left\{ \mathcal{S}_{X}^{\red{r}},\mathcal{S}_{X}^{\green{g}},\mathcal{S}_{X}^{\blue{b}},\mathcal{S}_{Z}^{\red{r}},\mathcal{S}_{Z}^{\green{g}},\mathcal{S}_{Z}^{\blue{b}}\right\}~,
\end{align}
where $\mathcal{S}_{X}^{\red{r}},\mathcal{S}_{X}^{\green{g}},\mathcal{S}_{X}^{\blue{b}}$ are supported at vertices of the lattice,
\begin{align}
    \mathcal{S}_{X}^{\red{r}} = \left\{S_{X;v}^{\red{r}}\right\}~, \quad S_{X;v}^{\red{r}} = \left(\prod_{\partial e\supset v} X^{\rd}_e\right) \prod_{e',e'':\int\tilde v\cup \tilde e'\cup \tilde e''\neq 0}CZ^{\gr,\bl}_{e',e''}~,
\end{align}
\begin{align}
    \mathcal{S}_{X}^{\gr} = \left\{S_{X;v}^{\gr}\right\}~, \quad S_{X;v}^{\gr} = \left(\prod_{\partial e\supset v} X^{\gr}_e\right) \prod_{e',e'': \int\tilde e''\cup \tilde v\cup \tilde e'\neq 0}
   CZ^{\bl,\rd}_{e,e'}~,
\end{align}
\begin{align}
    \mathcal{S}_{X}^{\bl} = \left\{S_{X;v}^{\bl}\right\}~, \quad S_{X;v}^{\bl} = \left(\prod_{\partial e\supset v} X^{\bl}_e\right)\prod_{e',e'':\int \tilde e'\cup \tilde e''\cup \tilde v\neq 0} CZ^{\rd,\gr}_{e',e''}~,
\end{align}
where $\tilde v$ is a 0-cochain that takes value 1 on vertex $v$ and 0 on other vertices, $\tilde e$ is a 1-cochain that takes value $1$ on edge $e$ and 0 on other edges. $CZ^{c,c'}_{e,e'}$ is the CZ gate for a pair of qubits with colors $c,c'$, on the edges $e,e'$, respectively. 

The stabilizers $\mathcal{S}_{Z}^{\red{r}},\mathcal{S}_{Z}^{\green{g}},\mathcal{S}_{Z}^{\blue{b}}$ are supported at faces,
\begin{align}
\mathcal{S}_{Z}^{\red{r}} = \left\{S_{Z;f}^{\red{r}}\right\}~,\quad 
    S^{\red{r}}_{Z;f} = \prod_{e\in \partial f} Z_e^{\red{r}}~,
\end{align}
\begin{align}
\mathcal{S}_{Z}^{\gr} = \left\{S_{Z;f}^{\gr}\right\}~,\quad 
    S^{\gr}_{Z;f} = \prod_{e\in \partial f} Z_e^{\gr}~,
\end{align}
\begin{align}
\mathcal{S}_{Z}^{\bl} = \left\{S_{Z;f}^{\bl}\right\}~,\quad 
    S^{\bl}_{Z;f} = \prod_{e\in \partial f} Z_e^{\bl}~.
\end{align}
As discussed in \cite{Hsin2024_non-Abelian}, the model is a non-commuting Clifford stabilizer model, i.e., $\mathcal{S}_{X}^{\red{r}},\mathcal{S}_{X}^{\green{g}},\mathcal{S}_{X}^{\blue{b}}$ are commutative within the stabilizer space of $\mathcal{S}_{Z}^{\red{r}},\mathcal{S}_{Z}^{\green{g}},\mathcal{S}_{Z}^{\blue{b}}$.

\change{As illustrated in the main text, we place the above Clifford stabilizer code on a triangle bounded by three gapped boundaries on each edge. Using the lattice $\Z_2^3$ gauge fields $a_{\rd}, a_{\gr}, a_{\bl}$, each gapped boundary $\mathcal{L}_{\rd}, \mathcal{L}_{\bl}, \mathcal{L}_{\rd\bl}$ is characterized by the boundary condition on gauge fields
\begin{align}
    \begin{split}
        \mathcal{L}_{\rd}: a_{\rd} = 0~, \quad \ \mathcal{L}_{\bl}: a_{\bl} = 0~, \quad 
        \mathcal{L}_{\rd\bl}: a_{\rd} + a_{\bl} = a_{\gr} = 0~. \\
    \end{split}
\end{align}
Each boundary condition is enforced using the $Z$ stabilizers at boundaries, and the stabilizers with boundaries are illustrated in Fig.~\ref{fig:logicalZ}. 
As discussed in the main text, this code on a triangle encodes a single logical qubit. Indeed, with the above boundary conditions, there is a single nontrivial configuration of $\Z_2^3$ gauge field on a triangle; $\int a_{\rd} = 1$ along the boundary $\mathcal{L}_{\bl}$, $\int a_{\bl} = 1$ along the boundary $\mathcal{L}_{\rd}$, $\int a_{\rd} = \int a_{\bl} = 1$ along the boundary $\mathcal{L}_{\rd\bl}$. This configuration corresponds to the logical state $\ket{\overline{1}}$, while the trivial $\Z_2^3$ gauge field corresponds to $\ket{\overline{0}}$.}

\change{The code supports a single transversal logical $\overline{Z}$ operator, which is generated by a junction of $Z$ string operators illustrated in Fig.~\ref{fig:logicalZ}.}

\begin{figure}[htb]
    \centering
    \includegraphics[width=0.8\linewidth]{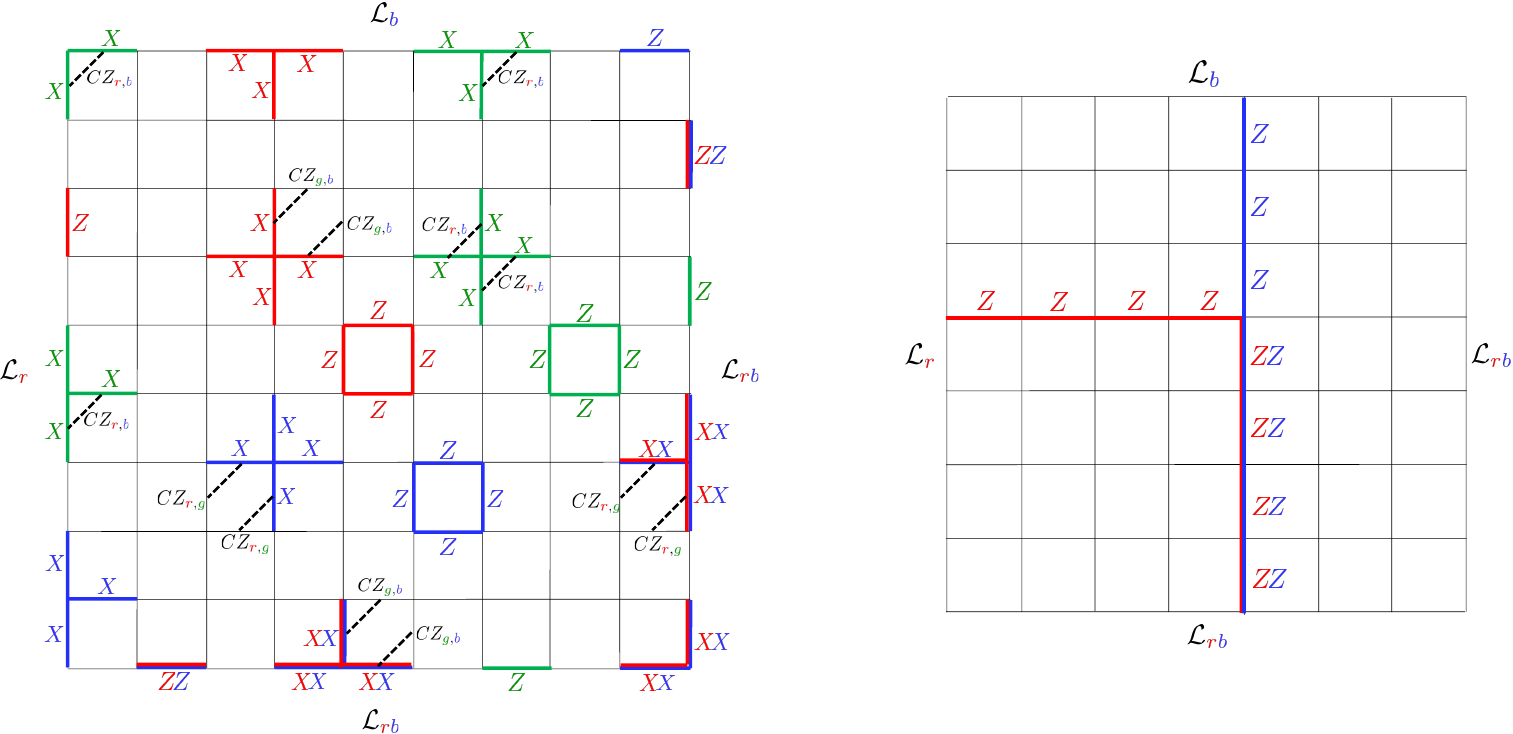}
    \caption{Left: the Clifford stabilizer code located on a square lattice with three gapped boundaries. Right: A single logical $Z$ operator of the code. 
    }
    \label{fig:logicalZ}
\end{figure}

\section{Emergent Automorphism Symmetry in 2D}
\label{app:automorphism}
The automorphism symmetry can be expressed as $U=WV$, where
\begin{align}
   V= \bigotimes_e \text{CNOT}^{(\rd,\gr)}_{e} \text{CNOT}^{(\bl,\gr)}_{e}~,
\end{align}
which transforms the Pauli operators according to the automorphism (\ref{eqn:automorphism}):
\begin{align}
\begin{split}
    & X_e^{\rd}\leftrightarrow X_e^{\rd}X_e^{\gr}, \quad X_e^{\gr}\leftrightarrow X_e^{\gr}, \quad X_e^{\bl}\leftrightarrow X_e^{\bl}X_e^{\gr}~, \\
    & Z_e^{\rd} \leftrightarrow Z_e^{\rd}, \quad Z_e^{\gr} \leftrightarrow Z_e^{\rd}Z_e^{\gr}Z_e^{\bl}, \quad Z_e^{\bl} \leftrightarrow Z_e^{\bl}~.
\end{split}
\end{align}
The automorphism symmetry is
\begin{equation}\label{eqn:automorphismsymmetryoperator app}
     U =  WV~, \quad 
     W = (-1)^{\int \frac{a_{\rd} \cup a_{\bl}}{2}}=\prod_{\Delta_{012}} (\text{CS}^{\rd,\bl}_{e_{01},e_{12}})^{\sigma(\Delta_{012})}~,
\end{equation}
where the product is over 2-simplices $\Delta_{012}$, and $\sigma(\Delta_{012})=\pm 1$ is the orientation of the 2-simplex. On a square lattice, this operator $U$ takes the form of \eqref{eqn:automorphismsymmetryoperator} in the main text.

\subsection{Automorphism symmetry in path integral}
\label{app:automorphism path integral}
One way to see the automorphism symmetry is a symmetry of the ground state code subspace is using the path integral formalism for the twisted $\mathbb{Z}_2^3$ gauge theory that describes the ground state subspace.

The path integral is
\begin{equation}\label{eqn:pathintegral1}
    Z[M]=\sum_{a_{\rd},a_{\gr},a_{\bl}\in H^1(M,\mathbb{Z}_2)} (-1)^{\int a_{\rd}\cup a_{\gr}\cup a_{\bl}}~. 
\end{equation}
where we labeled $\Z_2^3$ gauge fields by colors $\rd,\gr,\bl$ according to the main text.
Using Poincar\'e duality, we can also express the path integral in terms of membranes $m_{\rd},m_{\gr},m_{\bl}$
\begin{equation}\label{eqn:pathintegral2}
    Z[M]=\sum_{m_{\rd},m_{\gr},m_{\bl}\in H_2(M,\mathbb{Z}_2)} (-1)^{\#(m_{\rd},m_{\gr},m_{\bl})}~,
\end{equation}
where $\#(m_{\rd},m_{\gr},m_{\bl})$ is the triple intersection number of the three membranes $m_{\rd},m_{\gr},m_{\bl}$.

The automorphism symmetry can be expressed in terms of $a_{\rd},a_{\gr},a_{\bl}$ that take values in $\{0,1\}^3$ as
\begin{equation}
    (a_{\rd},a_{\gr},a_{\bl})\rightarrow (a_{\rd},a_{\rd}+a_{\gr}+a_{\bl},a_{\bl})~.
\end{equation}
Under the above transformation, the path integral (\ref{eqn:pathintegral1}) transforms into
\begin{equation}
    Z[M]\rightarrow \sum_{a_{\rd},a_{\gr},a_{\bl}\in H^1(M,\mathbb{Z}_2)} (-1)^{\int a_{\rd}\cup (a_{\rd}+a_{\gr}+a_{\bl})\cup a_{\bl}}~.
\end{equation}
In the exponent,
\begin{equation}
    a_{\rd}\cup (a_{\rd}+a_{\gr}+a_{\bl})\cup a_{\bl}=a_{\rd}\cup a_{\gr}\cup a_{\bl}+a_{\rd}\cup a_{\bl}^2+a_{\rd}^2\cup a_{\bl}~.
\end{equation}
Using the identity $a\cup a = \frac{d\tilde a}{2}$ with a $\Z_4$ lift $\tilde a$ of $a$ (valid for $\Z_2$ 1-cocycles $a$), 
we can rewrite the right hand side as
\begin{equation}
    a_{\rd}\cup a_{\gr}\cup a_{\bl}+\frac{1}{2}d\left(\tilde a_{\rd} \cup \tilde a_{\bl}\right)~.
\end{equation}
In other words, the transformation leaves the exponent invariant up to a total derivative. Since the total derivative integrates to zero for closed $M$, the path integral transforms as
\begin{equation}\label{eqn:transformedZ}
    Z[M]\rightarrow \sum_{a_1,a_2,a_3\in H^1(M,\mathbb{Z}_2)} (-1)^{\int a_{\rd}\cup a_{\gr}\cup a_{\bl}}\cdot  i^{\int_M d\left(\tilde a_{\rd} \cup \tilde a_{\bl}\right)}=Z[M]~,
\end{equation}
where we used $\int_M d\left(\tilde a_{\rd} \cup \tilde a_{\bl}\right)=0$ for closed $M$.
Thus we conclude that the path integral is invariant under the automorphism symmetry.

Moreover, when $M$ is inserted with domain wall $D$ of the symmetry, i.e. we only perform transformation on half spacetime $[0,\infty)_t$ with boundary $D$ at $t=0$, the total derivative in (\ref{eqn:transformedZ}) picks up the following contribution in the path integral
\begin{equation}\label{eqn:gaugedSPTsymm}
  W =  i^{\int_D \tilde a_{\rd} \cup \tilde a_{\bl}}~.
\end{equation}
In other words, the automorphism symmetry inserts the gauged SPT operator (\ref{eqn:gaugedSPTsymm}) in the path integral. The gauged SPT operator by itself is not topological, and it is only topological and generates symmetry when accompanied by the automorphism transformation. The operator \eqref{eqn:gaugedSPTsymm} corresponds to the additional operator $W$ in \eqref{eqn:automorphismsymmetryoperator app} for the automorphism symmetry. See \cite{automorphism2025} for a general discussion of automorphism symmetry in twisted gauge theories, where it is shown that the automorphism symmetry can become extended or a higher-group/non-invertible symmetry.

\subsubsection{Automorphism permutes topological excitations}
\label{app:permutation}
Here we will study how the automorphism symmetry $U$ acts on the particle excitations, i.e. detectable errors in the Clifford stabilizer codes.
We will present two derivations: one uses how stabilizers are permuted under the symmetry, the other uses path integral formalism of the ground state subspace.

The electric charge excitations $e_{\rd},e_{\gr},e_{\bl}$ violate the $X$ stabilizers $S^{\rd}_X,S^{\gr}_X,S^{\bl}_X$ while the magnetic charge excitations $m_{\rd},m_{\gr},m_{\bl}$ violate the $Z$ stabilizers $S^{\rd}_Z,S^{\gr}_Z,S^{\bl}_Z$. When there are more than one stabilizers that are violated, the excitations have subscript with the corresponding colors, such as $e_{\rd\gr}$ being the violation of both $S^{\rd}_X,S^{\gr}_X$.
By checking how the stabilizers are permuted under the symmetry, we can track how the excitations are permuted. From the transformation of the stabilizers (\ref{eqn:permutingstabilizers}), we conclude the mapping of excitations under the automorphism symmetry:
\begin{align}
    &m_{\rd}\leftrightarrow m_{\rd \gr},\; m_{\gr}\leftrightarrow m_{\gr}, \; m_{\bl}\leftrightarrow m_{\gr\bl}, \; m_{\rd\bl}\leftrightarrow m_{\rd\bl}~,\cr
    &e_{\rd}\leftrightarrow e_{\rd},\quad e_{\gr}\leftrightarrow e_{\rd\gr\bl},\quad  e_{\bl}\leftrightarrow e_{\bl}~.
    \label{eq:anyonpermutationapp}
\end{align}
For example, violation of $S^{\gr}_X$ corresponds to the violation of all $S^{\rd}_X,S^{\gr}_X,S^{\bl}_X$ under the automorphism symmetry, and thus the excitation $e_{\gr}$ corresponds to $e_{\rd\gr\bl}$ under the symmetry.

An alternative derivation uses the path integral of continuum field theory that describes the ground states, where there are states related by closed loop operators that create and annihilate a pair of electric and/or magnetic charge excitations. 
The field theory is described by the topological action
\begin{align}
    \pi\int a_{\rd}\cup a_{\gr} \cup a_{\bl}~,
\end{align}
where $a_{\rd},a_{\gr},a_{\bl}=0,1$ are $\Z_2$ gauge fields correspond to the $\rd,\gr,\bl$ types of qubits. One can embed these $\Z_2$ gauge fields into $U(1)$ as
\begin{equation}
    a_{\rd}\rightarrow a_{\rd}/\pi,\quad 
    a_{\gr}\rightarrow a_{\gr}/\pi,\quad a_{\bl}\rightarrow a_{\bl}/\pi~,
\end{equation}
where the gauge fields $a_{\rd},a_{\gr},a_{\bl}$ now have holonomies in $0,\pi$ instead of $0,1$ because of the rescaling.
We can introduce Lagrangian multipliers $\tilde a_{\rd},\tilde a_{\gr},\tilde a_{\bl}$ to enforce these constraints on the holonomies via a BF type coupling \cite{Horowitz:1989ng,Maldacena:2001ss,Kapustin:2014gua}.
In terms of the rescaled gauge fields and the Lagrangian multipliers, the theory can be expressed as
\begin{equation}\label{eqn:U(1)norm}
    \frac{1}{\pi^2}\int a_{\rd} a_{\gr} a_{\bl}+\frac{2}{2\pi}\int (a_{\rd} d{b}_{\rd} + a_{\gr} d{b}_{\gr} +  a_{\bl} d{b}_{\bl})~.
\end{equation}
In the theory, the gauge invariant operators $e^{i\int a_{\rd}},e^{i\int a_{\gr}},e^{i\int a_{\bl}}$ on closed loops create and annihilate a pair of electric charges $e_{\rd},e_{\gr},e_{\bl}$, and similarly 
the operators $e^{i\int {b}_{\rd}},e^{i\int {b}_{\gr}},e^{i\int {b}_{\bl}}$ create and annihilate a pair of magnetic charges $m_{\rd},m_{\gr},m_{\bl}$.
We can then identify the permutation (\ref{eq:anyonpermutationapp}) as the automorphism of $\Z_2^3$ gauge group which transforms the gauge fields as
\begin{align}
    a_{\rd}\to a_{\rd}, \quad a_{\gr}\to a_{\rd}+a_{\gr}+a_{\bl}, \quad a_{\bl}\to a_{\bl}~,
\end{align}
and
\begin{align}
    {b}_{\rd} \to {b}_{\rd} + {b}_{\gr}, \quad {b}_{\gr} \to {b}_{\gr}, \quad {b}_{\bl} \to {b}_{\bl}+{b}_{\gr}~.
\end{align}
Thus we find that the transformation of the gauge fields indeed corresponds to the permutation of topological excitations in (\ref{eq:anyonpermutationapp}).

\subsubsection{Generalization to higher dimensions}

The model can be generalized to $(N-1)$ spatial dimensions, where the gauge theory is a twisted $\mathbb{Z}_2^N$ gauge theory, with the path integral
\begin{equation}\label{eqn:pathintegral}
    Z[M]=\sum_{a_1,a_2,\cdots, a_N\in H^1(M,\mathbb{Z}_2)} (-1)^{\int a_1\cup a_2\cup \cdots\cup a_N}~,
\end{equation}
where $M$ is the spacetime manifold. We have omitted an overall normalization that removes the gauge transformations. The one-cocycles (gauge fields) $a_1,\cdots, a_N$ are $\mathbb{Z}_2^N$ valued. 
We can rewrite the path integral using Poincar\'e duality, where the one-cocycles are replaced by $(N-1)$-cycles $m_1,m_2,\cdots,m_N$ that take value in $\mathbb{Z}_2^N$:
\begin{equation}\label{eqn:pathintegral'}
    Z[M]=\sum_{m_1,m_2,\cdots, m_N\in H_{N-1}(M,\mathbb{Z}_2)} (-1)^{\#\left(m_1,\cdots,m_N\right)}~,
\end{equation}
where $\#\left(m_1,\cdots,m_N\right)$ is the intersection number of the $(N-1)$-cycles $m_1,\cdots,m_N$.

Consider the automorphism of $\mathbb{Z}_2^N$ gauge group that acts on the gauge fields as
\begin{equation}
a_1\rightarrow a_1+a_2+a_3+a_4+\cdots+a_N~.    
\end{equation}
The weight of the path integral changes as
\begin{equation}
    \omega=\pi a_1\cup a_2\cup\cdots \cup a_N\rightarrow \omega+d\alpha~,
\end{equation}
where
\begin{equation}
    \alpha=\frac{\pi}{2}\left(a_2\cup a_3\cup \cdots a_N\right)+\pi\sum_{i=3}^N\sum_{j=2}^{i-1}a_2\cup\cdots a_{j-1}\cup\left(a_i\cup_1 a_j\right)\cup a_{j+1}\cdots\cup a_i\cdots \cup a_N~.
\end{equation}
In the summation, the first few terms are $\pi (a_3\cup_1 a_2)\cup a_3\cup\cdots a_N+\pi (a_4\cup_1 a_2)\cup a_3\cup a_4\cdots a_N+\pi a_2\cup (a_4\cup_1 a_3)\cup a_4\cdots a_N$.
Therefore, the automorphism symmetry is generated by the operator in the form of $U=WV$, where $V$ is a transversal CNOT gate inducing the permutation of gauge fields, and and $W=e^{i\int\alpha}$ with the integral over the whole $N$-dimensional space.

\section{Details of logical T gates in 2D}
\label{app:Tgate}

\subsection{Proof of Theorem 1: emergent automorphism symmetry in Clifford stabilizer}
\label{sec:proof of theorem 1}

Here we show the following Theorem 1 presented in the main text:
\begin{quote}
{\bf Theorem 1.} 
The unitary operator (\ref{eqn:automorphismsymmetryoperator}) preserves the logical subspace and thus is an emergent symmetry of the Clifford stabilizer model.
\end{quote}

\begin{proof}
   First, let us check that $U$ preserves the $Z$-stabilizer subspace of $\mathcal{S}_{Z}^{\red{r}}, \mathcal{S}_{Z}^{\green{g}},\mathcal{S}_{Z}^{\blue{b}}$.
   Conjugating by $U$ acts on the $Z$-stabilizers as
   \begin{align}
       \mathcal{S}_{Z}^{\rd} \to \mathcal{S}_{Z}^{\rd}~, \quad \mathcal{S}_{Z}^{\gr} \to \mathcal{S}_{Z}^{\rd}\mathcal{S}_{Z}^{\gr}~, \quad \mathcal{S}_{Z}^{\bl} \to \mathcal{S}_{Z}^{\bl}~,
\end{align}
therefore preserves the $Z$-stabilizer subspace.

The remaining task is to show that $U$ induces the automorphism of $X$-stabilizers $\mathcal{S}_{X}^{\red{r}},\mathcal{S}_{X}^{\green{g}},\mathcal{S}_{X}^{\blue{b}}$ within the $Z$-stabilizer subspace, therefore preserves the logical subspace.
From now, let us restrict to the $Z$-stabilizer subspace $\mathcal{S}_{Z}^{\red{r}}= \mathcal{S}_{Z}^{\green{g}}=\mathcal{S}_{Z}^{\blue{b}}=1$. Using the $\Z_2$ gauge fields $Z^{\red{r}}=(-1)^{a_{\red{r}}},Z^{\green{g}}=(-1)^{a_{\green{g}}},Z^{\blue{b}}=(-1)^{a_{\blue{b}}}$, the $Z$-stabilizer states are characterized by those satisfying $da_{\rd}=da_{\gr}=da_{\bl}=0$ mod 2, i.e., states with flat $\Z_2^3$ gauge fields.

Now $V$ transforms the $X$-stabilizers as
\begin{align}
    \begin{split}
        V S_{X;v}^{\rd} V^\dagger &=  \left(\prod_{v\subset \partial e} X_e^{\rd}X_e^{\gr}\right)(-1)^{\int \hat v\cup (a_{\rd}+a_{\gr}+a_{\bl})\cup a_{\bl}}~, \\
        V S_{X;v}^{\gr} V^\dagger &=\left(\prod_{v\subset \partial e} X^{\gr}_{e}\right)(-1)^{\int a_{\rd}\cup \hat v\cup  a_{\bl}}~, \\
        V S_{X;v}^{\bl} V^\dagger &=  \left(\prod_{v\subset \partial e} X_e^{\gr}X_e^{\bl}\right)(-1)^{\int a_{\rd}\cup(a_{\rd}+a_{\gr}+a_{\bl})\cup \hat v}~. \\
    \end{split}
\end{align}
With $da_{\rd}=da_{\gr}=da_{\bl}=0$ in mind, the commutator between the Pauli $X$ terms and $W$ is given by   
\begin{align}
\begin{split}
&\left(\prod_{v\subset \partial e} X_e^{\rd}\right)W \left(\prod_{v\subset \partial e} X_e^{\rd}\right)W^\dagger \\
&\quad = \exp\left( \pi i \int \frac{(\widetilde{a_{\rd}+d\hat{v}})\cup \tilde a_{\bl}}{2} -\frac{\tilde a_{\rd}\cup \tilde a_{\bl}}{2}\right)
= \exp\left(\pi i\int \frac{\widetilde{d\hat{v}}\cup \tilde a_{\bl}}{2} + (d\hat v\cup_1 a_{\rd})\cup a_{\bl}\right) \\
&\quad = \exp\left(\pi i\int \frac{\widetilde{d\hat{v}}\cup \tilde a_{\bl}}{2} + d\hat v\cup(a_{\bl}\cup_1 a_{\rd})+ (d\hat{v}\cup a_{\bl})\cup_1 a_{\rd}\right)  
\\
&\quad 
= \exp\left(\pi i\int \frac{d\tilde{\hat{v}}\cup \tilde a_{\bl}}{2} + \hat{v}\cup d\hat{v}\cup a_{\bl} + d\hat v\cup(a_{\bl}\cup_1 a_{\rd})+ (d\hat{v}\cup a_{\bl})\cup_1 a_{\rd}\right)  \\
&\quad = \exp\left(\pi i\int \hat v\cup a_{\bl}\cup a_{\bl} + \hat{v}\cup d\hat{v}\cup a_{\bl} + \hat v\cup(a_{\bl}\cup a_{\rd}+a_{\rd}\cup a_{\bl})+ d(\hat{v}\cup a_{\bl})\cup_1 a_{\rd}\right) \\
&\quad = \exp\left(\pi i\int \hat v\cup a_{\bl}\cup a_{\bl} + \hat{v}\cup d\hat{v}\cup a_{\bl} + \hat v\cup a_{\rd}\cup a_{\bl}+ a_{\rd}\cup \hat{v}\cup a_{\bl}\right)~, \\
&\left(\prod_{v\subset \partial e} X_e^{\gr}\right)W \left(\prod_{v\subset \partial e} X_e^{\gr}\right)W^\dagger  = 1~, \\
&\left(\prod_{v\subset \partial e} X_e^{\bl}\right)W \left(\prod_{v\subset \partial e} X_e^{\bl}\right)W^\dagger \cr &\quad = \exp\left(\pi i\int \frac{\tilde{a}_{\rd}\cup (\widetilde{a_{\bl}+d\hat{v}}) }{2} -\frac{\tilde a_{\rd}\cup \tilde a_{\bl}}{2}\right) 
= \exp\left(\pi i\int \frac{\tilde a_{\rd}\cup\widetilde{d\hat{v}}}{2} + a_{\rd}\cup (d\hat v\cup_1 a_{\bl})\right) \\
&\quad = \exp\left(\pi i\int \frac{\tilde a_{\rd}\cup\widetilde{d\hat{v}}}{2}  + (a_{\rd}\cup_1 a_{\bl})\cup d\hat v + (a_{\rd}\cup d\hat v )\cup_1 a_{\bl}\right) 
\\
&\quad
= \exp\left(\pi i\int \frac{\tilde a_{\rd}\cup d\tilde{\hat{v}}}{2}  + a_{\rd}\cup \hat{v}\cup d\hat{v} + (a_{\rd}\cup_1 a_{\bl})\cup d\hat v + (a_{\rd}\cup d\hat v )\cup_1 a_{\bl}\right)  \\
&\quad = \exp\left(\pi i\int a_{\rd}\cup a_{\rd} \cup \hat{v}  + a_{\rd}\cup \hat{v}\cup d\hat{v}+ (a_{\rd}\cup a_{\bl}+a_{\bl}\cup a_{\rd})\cup \hat v + d(a_{\rd}\cup \hat v )\cup_1 a_{\bl}\right)  \\
&\quad = \exp\left(\pi i\int a_{\rd}\cup a_{\rd} \cup \hat{v}  + a_{\rd}\cup d\hat{v}\cup\hat{v}+ a_{\rd}\cup a_{\bl}\cup \hat v + a_{\rd}\cup \hat v \cup a_{\bl}\right)~, \\
\end{split}
\end{align}
where we used the Hirsch identity $(a\cup b)\cup_1 c = a\cup(b\cup_1 c)+ (a\cup_1 c)\cup b$ for $\Z_2$ 1-cochains. We also used $\widetilde{d\hat v} = d\tilde{\hat v} + 2\hat{v}\cup d\hat{v}$ mod 4,
which can be explicitly verified using the definition of $\tilde v$ as follows: for an edge $(01)$ with $\tilde v(0),\tilde v(1)$ taking the possible values $(0,0),(0,1),(1,0)$, the left hand side is $\widetilde{d\hat v}=0,1,1$ for the 3 cases respectively, while the right hand side $d\tilde{\hat v} + 2\hat{v}\cup d\hat{v}$ equals $0+0=0,1+0=1,-1+2\cdot(-1)=-3$ respectively, in agreement with the left hand side mod 4.

Using the above actions of $W$, $U=WV$ transforms the stabilizers by
\begin{align}\label{eqn:permutingstabilizers}
    \begin{split}
        WV S_{X;v}^{\rd} V^\dagger W^\dagger &=   \left(\prod_{v\subset \partial e} X_e^{\rd}X_e^{\gr}\right)(-1)^{\int \hat v\cup (a_{\gr}+d\hat{v})\cup a_{\bl} + a_{\rd} \cup \hat{v}\cup a_{\bl}} = S_{X;v}^{\rd}S_{X;v}^{\gr}~, \\
        WV S_{X;v}^{\gr} V^\dagger W^\dagger &=\left(\prod_{v\subset \partial e} X^{\gr}_{e}\right)(-1)^{\int a_{\rd}\cup \hat v\cup  a_{\bl}} =S_{X;v}^{\gr}~, \\
        WV S_{X;v}^{\bl} V^\dagger W^\dagger &=  \left(\prod_{v\subset \partial e} X_e^{\bl}X_e^{\gr}\right)(-1)^{\int a_{\rd}\cup (a_{\gr}+d\hat{v})\cup \hat v + a_{\rd} \cup \hat{v}\cup a_{\bl}}= S_{X;v}^{\gr}S_{X;v}^{\bl}~. \\
    \end{split}
\end{align}
This implies that $U=VW$ induces an automorphism of stabilizers,  therefore generates a logical gate.

\end{proof}

\subsection{proof of Theorem 2: Automorphism symmetry with boundaries}
\label{app:Tgateboundary}
\subsubsection{Automorphism symmetry with boundaries}
\label{app:hinge}
As described in Sec.~\ref{app:automorphism}, the automorphism symmetry of the stabilizer code in $N$ spatial dimensions has the form of $U=WV$ with $W$ given by an integral of a $N$-cochain, in the form of
\begin{align}
    W=e^{i\int_{\text{space}}\alpha}
\end{align}
with $\alpha$ a $N$-cochain expressed by gauge fields. In the presence of gapped boundaries, the definition of $W$ has to be properly modified to behave as a logical gate. Such boundary modifications of the logical gates given by integral of cochains are discussed in \cite{Hsin2024:classifying}. Here we describe the constructions of such logical gates in the presence of boundaries following \cite{Hsin2024:classifying}.

At the gapped boundary, the $\Z_2$ gauge fields $a_1,\dots, a_N$ are subject to specific boundary conditions. To construct a logical gate in the presence of the boundary, we require that the cochain $\alpha$ in the bulk trivializes at the boundary as
\begin{align}
    \alpha|_{\text{bdry}} = d\beta|_{\text{bdry}}~,
\end{align}
with a $(N-1)$-cochain $\beta$ at the boundary. The logical gate of the bulk-boundary system is then given by $U=WV$ with
\begin{align}
    W= \exp\left(i\int_{\text{bulk}}\alpha - i\int_{\text{bdry}}\beta\right)~.
\end{align}
Such a construction of logical gates is valid even when the bulk is bounded by multiple gapped boundary conditions, and the distinct boundary conditions are separated by a $(N-2)$-dimensional hinge. For instance, let us consider a setup where two gapped boundaries meet at a hinge, as shown in Fig.~\ref{fig:hinge}.
At the two boundaries (labeled by bdry, bdry' respectively) we have the boundary actions $\beta, \beta'$. The hinge then needs to trivialize the boundary actions when restricted,
\begin{align}
    (\beta - \beta')|_{\text{hinge}} = d\gamma|_{\text{hinge}}~,
\end{align}
with a $(N-2)$-cochain $\gamma$. The logical operator of the whole system is then given by $U=WV$ with
\begin{align}
    W= \exp\left(i\int_{\text{bulk}}\alpha - i\int_{\text{bdry}}\beta -i\int_{\text{bdry'}}\beta' + i\int_{\text{hinge}}\gamma\right)~.
    \label{eq:hingeaction}
\end{align}

\begin{figure}[htb]
    \centering
    \includegraphics[width=0.4\linewidth]{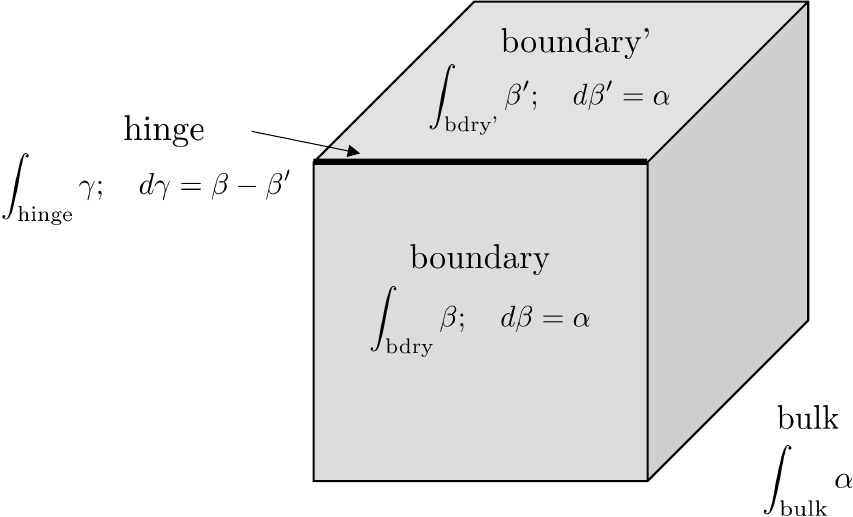}
    \caption{The operator $W$ is given by an integral of a cochain $\alpha$ in the bulk. There are a pair of gapped boundaries which trivializes the bulk cochain by $d\beta=\alpha, d\beta'=\alpha$ respectively. There is also a hinge between boundaries that trivializes the boundary action by $d\gamma=\beta-\beta'$. The logical operator in the whole space is given by $U=WV$ with $W$ given by \eqref{eq:hingeaction}.
    }
    \label{fig:hinge}
\end{figure}

\subsubsection{Proof of Theorem 2}

Now we show the following Theorem 2 in the main text:
\begin{quote}
{\bf Theorem 2.}
    The automorphism symmetry $U$ in the Clifford stabilizer code with the boundary condition implements logical T gate.
\end{quote}

\begin{proof}
To show this, let us evaluate the operator $W$ in the logical gate $U=WV$ in the presence of the boundary conditions. This can be achieved by modifying the operator $W$ at boundaries of a triangle (Fig.~\ref{fig:triangleboundary}) as reviewed above; the bulk 2-cochain $\alpha=\frac{\pi}{2} \tilde a_{\rd} \cup \tilde a_{\bl}$ gets trivialized under the boundary condition $\mathcal{L}_{\rd\bl}$,
where $a=a_{\rd}|=a_{\bl}|$ is the common value at the boundary condition. At this boundary $\mathcal{L}_{\rd\bl}$, we get $\alpha=d\beta$ with $\beta=\frac{\pi}{4}\tilde a$. 
This boundary condition $a_{\rd}|=a_{\bl}|$ arises due to the composite $e_{\rd\bl}$ condensation on the boundary. Such constraint on gauge fields is enforced by the boundary stabilizers $Z^{\rd}Z^{\bl}$ at the boundary ${\cal L}_{\rd\bl}$ (see Fig.~\ref{fig:trianglehamiltonian}).

The automorphism $U$ on the code is given by $U=WV$ with 
\begin{equation}
    W= \exp\left(\pi i \int_{\text{bulk}} \frac{\tilde a_{\rd} \cup \tilde a_{\bl}}{2} - \pi i \int_{\text{bdry}_{\rd\bl}}\frac{\tilde a_{\rd}}{4}\right)~,
\end{equation}
where we use tilde to denote the lift of $\mathbb{Z}_2$ to $\mathbb{Z}$ as the values $\{0,1\}$. 
To explicitly verify that $U$ preserves the code subspace, we just have to check the action of $U$ on boundary $X$-stabilizer, which is a combination $S^{\rd}_X S^{\bl}_X$ on boundary vertices (see Fig.~\ref{fig:trianglehamiltonian} in the main text); the other stabilizers in the bulk have been checked to be preserved in Sec.~\ref{sec:proof of theorem 1}.

In the computation below, suppose that $v$ is a vertex at the boundary $\mathcal{L}_{\rd\bl}$. First, $V$ transforms the stabilizer as
\begin{align}
    \begin{split}
        V S_{X;v}^{\rd}S_{X;v}^{\bl}V^\dagger &=  \left(\prod_{v\subset \partial e} X_e^{\rd}X_e^{\bl}\right)(-1)^{\int \hat v\cup (a_{\rd}+a_{\gr}+a_{\bl})\cup a_{\bl} + a_{\rd}\cup(a_{\rd}+a_{\gr}+a_{\bl})\cup \hat v} \\
        &= S_{X;v}^{\rd}S_{X;v}^{\bl}(-1)^{\int \hat v\cup (a_{\rd}+a_{\bl})\cup a_{\bl} + a_{\rd}\cup(a_{\rd}+a_{\bl})\cup \hat v}~.
        \label{commutation with V boundary}
    \end{split}
\end{align}
$W$ acts on the stabilizer through the commutation with $X$-terms, which we will evaluate below. During the computations we use the following relations (some of them are only valid within the $Z$-stabilizer subspace):
\begin{align}
    \begin{split}
    \widetilde{a+d\hat v} &=  \tilde a + \widetilde{d\hat{v}}-2\tilde a\cup_1\widetilde{d\hat v}~,\\
        \widetilde{d\hat{v}} &= d\tilde{\hat{v}} - 2\tilde{\hat{v}}\cup d\tilde{\hat{v}}~, \\
        \frac{d\tilde a}{2} &= a\cup a \quad \mod 2~, \\
        a_{\rd} &= a_{\bl} \quad \text{at the boundary,} \\
        (a\cup b)\cup_1 c &= a\cup(b\cup_1 c)+ (a\cup_1 c)\cup b \quad \text{for $\Z_2$ 1-cochains,} \\
        -\tilde a \tilde{\hat v} + \tilde{\hat v}\tilde a + \tilde a \cup_1 d\tilde{\hat v} &= 0~, \\
        \hat{v} ad\hat v &= d\hat{v}a\hat v = 0~, \\
        \hat{v}((a_{\bl}+d\hat{v})\cup_1 a_{\rd}) &= (\hat v\cup(a_{\bl}+d\hat v))\cup_1 a_{\rd}~,\\
        ((a_{\rd}+d\hat{v})\cup_1 a_{\bl})\hat{v} &= ((a_{\rd}+d\hat v)\cup \hat{v})\cup_1 a_{\bl} ~, \\   
        \hat{v}a_{\rd} &= a_{\rd}\cup_1({\hat{v}}d{\hat v})\quad  \mod 2~.
    \end{split}
\end{align}
Now we have
\begin{align}
    \begin{split}
        &\left(\prod_{v\subset \partial e} X_e^{\rd}X_e^{\bl}\right)W \left(\prod_{v\subset \partial e} X_e^{\rd}X_e^{\bl}\right)W^\dagger \\
        &= \exp\left(\pi i \int_{\text{bulk}} \frac{(\widetilde{ a_{\rd}+d\hat{v}})(\widetilde{ a_{\bl}+d\hat{v}}) - \tilde{a}_{\rd} \tilde{a}_{\bl}}{2}\right)
        \exp\left(-\pi i\int_{\text{bdry}}\frac{(\widetilde{a_{\rd}+d\hat{v}})-\tilde{a}_{\rd}}{4}\right) \\
        &= \exp\left(\pi i \int_{\text{bulk}} \frac{\tilde a_{\rd}d\tilde{\hat v}+d\tilde{\hat v}\tilde a_{\bl} +d\tilde{\hat v}d\tilde{\hat v}}{2} + (a_{\rd}\cup_1d\hat{v}+\hat{v}d\hat{v})(a_{\bl}+d\hat{v})+(a_{\rd}+d\hat{v})(a_{\bl}\cup_1d\hat{v}+\hat{v}d\hat{v})
    \right) \\
    &\quad \times \exp\left(\pi i \int_{\text{bdry}}\frac{-d\tilde{\hat{v}} +2\tilde{\hat{v}}d\tilde{\hat v}+2\tilde a_{\rd}\cup_1(d\tilde{\hat{v}}-2\tilde{\hat{v}}d\tilde{\hat v}) }{4} \right) \\
    &= \exp\left(\pi i \int_{\text{bulk}} a_{\rd}a_{\rd}{\hat v}+{\hat v}a_{\bl}a_{\bl} + (a_{\rd}\cup_1d\hat{v}+\hat{v}d\hat{v})(a_{\bl}+d\hat{v})+(a_{\rd}+d\hat{v})(a_{\bl}\cup_1d\hat{v}+\hat{v}d\hat{v})
    \right) \\
    &\quad \times \exp\left(\pi i \int_{\text{bdry}}\frac{-\tilde a_{\rd}\tilde{\hat v} + \tilde{\hat{v}}\tilde{a}_{\bl} +2\tilde{\hat{v}}d\tilde{\hat v}+\tilde a_{\rd}\cup_1(d\tilde{\hat{v}}-2\tilde{\hat{v}}d\tilde{\hat v}) }{2} \right) \\
     &= \exp\left(\pi i \int_{\text{bulk}} a_{\rd}a_{\rd}{\hat v}+{\hat v}a_{\bl}a_{\bl}  + \hat vd\hat v (a_{\bl}+d\hat{v}) + d\hat v ((a_{\bl}+d\hat{v})\cup_1 a_{\rd} )+ (d\hat v\cup (a_{\bl}+d\hat{v}))\cup_1 a_{\rd}
    \right) \\
    &\quad \times  \exp\left(\pi i \int_{\text{bulk}} (a_{\rd}+d\hat{v})\hat vd\hat v + ((a_{\rd}+d\hat{v})\cup_1 a_{\bl} ) d\hat v + ((a_{\rd}+d\hat{v})\cup d\hat{v}) \cup_1 a_{\bl}   \right) \\
    &\quad \times \exp\left(\pi i \int_{\text{bdry}}\frac{-\tilde a_{\rd}\tilde{\hat v} + \tilde{\hat{v}}\tilde{a}_{\bl} +2\tilde{\hat{v}}d\tilde{\hat v}+\tilde a_{\rd}\cup_1(d\tilde{\hat{v}}-2\tilde{\hat{v}}d\tilde{\hat v}) }{2} \right) \\
    &= \exp\left(\pi i \int_{\text{bulk}} a_{\rd}a_{\rd}{\hat v}+{\hat v}a_{\bl}a_{\bl}  + \hat vd\hat v (a_{\bl}+d\hat{v}) + \hat v ((a_{\bl}+d\hat{v}) a_{\rd} +a_{\rd}(a_{\bl}+d\hat{v}))+ (d\hat v\cup (a_{\bl}+d\hat{v}))\cup_1 a_{\rd}
    \right) \\
    &\quad \times  \exp\left(\pi i \int_{\text{bulk}} (a_{\rd}+d\hat{v})\hat vd\hat v + ((a_{\rd}+d\hat{v}) a_{\bl}+a_{\bl}(a_{\rd}+d\hat{v}) ) \hat v + ((a_{\rd}+d\hat{v})\cup d\hat{v}) \cup_1 a_{\bl}   \right) \\
    &\quad \times \exp\left(\pi i \int_{\text{bdry}}\frac{-\tilde a_{\rd}\tilde{\hat v} + \tilde{\hat{v}}\tilde{a}_{\bl} +2\tilde{\hat{v}}d\tilde{\hat v}+\tilde a_{\rd}\cup_1(d\tilde{\hat{v}}-2\tilde{\hat{v}}d\tilde{\hat v}) }{2} + \hat{v}((a_{\bl}+d\hat{v})\cup_1 a_{\rd}) +((a_{\rd}+d\hat{v})\cup_1 a_{\bl})\hat{v}\right) \\
    &= \exp\left(\pi i \int_{\text{bulk}} a_{\rd}a_{\rd}{\hat v}+{\hat v}a_{\bl}a_{\bl}  + \hat vd\hat v (a_{\bl}+d\hat{v}) + \hat v ((a_{\bl}+d\hat{v}) a_{\rd} +a_{\rd}(a_{\bl}+d\hat{v}))+ (d\hat v\cup (a_{\bl}+d\hat{v}))\cup_1 a_{\rd}
    \right) \\
    &\quad \times  \exp\left(\pi i \int_{\text{bulk}} (a_{\rd}+d\hat{v})\hat vd\hat v + ((a_{\rd}+d\hat{v}) a_{\bl}+a_{\bl}(a_{\rd}+d\hat{v}) ) \hat v + ((a_{\rd}+d\hat{v})\cup d\hat{v}) \cup_1 a_{\bl}   \right) \\
    &\quad \times \exp\left(\pi i \int_{\text{bdry}}\tilde{\hat{v}}d\tilde{\hat v}+\tilde a_{\rd}\cup_1(\tilde{\hat{v}}d\tilde{\hat v})  + \hat{v}((a_{\bl}+d\hat{v})\cup_1 a_{\rd}) +((a_{\rd}+d\hat{v})\cup_1 a_{\bl})\hat{v}\right) \\
     &= \exp\left(\pi i \int_{\text{bulk}} a_{\rd}a_{\rd}{\hat v}+{\hat v}a_{\bl}a_{\bl}  + \hat vd\hat v (a_{\bl}+d\hat{v}) + \hat v ((a_{\bl}+d\hat{v}) a_{\rd} +a_{\rd}(a_{\bl}+d\hat{v}))+ (\hat v (a_{\bl}+d\hat{v}))a_{\rd} +a_{\rd} (\hat v (a_{\bl}+d\hat{v}))   \right) \\
    &\quad \times  \exp\left(\pi i \int_{\text{bulk}} (a_{\rd}+d\hat{v})\hat vd\hat v + ((a_{\rd}+d\hat{v}) a_{\bl}+a_{\bl}(a_{\rd}+d\hat{v}) ) \hat v + ((a_{\rd}+d\hat{v})\hat{v}) a_{\bl} + a_{\bl}((a_{\rd}+d\hat{v})\hat{v})\right) \\
    &\quad \times \exp\left(\pi i \int_{\text{bdry}}\tilde{\hat{v}}d\tilde{\hat v}+\tilde a_{\rd}\cup_1(\tilde{\hat{v}}d\tilde{\hat v})  + \hat{v}((a_{\bl}+d\hat{v})\cup_1 a_{\rd}) +((a_{\rd}+d\hat{v})\cup_1 a_{\bl})\hat{v} + (\hat v\cup(a_{\bl}+d\hat v))\cup_1 a_{\rd} + ((a_{\rd}+d\hat v)\cup \hat{v})\cup_1 a_{\bl}  \right) \\
     &= \exp\left(\pi i \int_{\text{bulk}} a_{\rd}a_{\rd}{\hat v}+{\hat v}a_{\bl}a_{\bl}  + \hat vd\hat v (a_{\bl}+d\hat{v}) + \hat v a_{\rd}(a_{\bl}+d\hat{v})+ d\hat{v}\hat vd\hat v + (a_{\rd}+d\hat{v}) a_{\bl}) \hat v + d\hat{v}\hat{v} a_{\bl}\right) \\
    &\quad \times \exp\left(\pi i \int_{\text{bdry}}\tilde{\hat{v}}d\tilde{\hat v}+\tilde a_{\rd}\cup_1(\tilde{\hat{v}}d\tilde{\hat v})  + \hat{v}((a_{\bl}+d\hat{v})\cup_1 a_{\rd}) +((a_{\rd}+d\hat{v})\cup_1 a_{\bl})\hat{v} + (\hat v\cup(a_{\bl}+d\hat v))\cup_1 a_{\rd} + ((a_{\rd}+d\hat v)\cup \hat{v})\cup_1 a_{\bl}  \right) \\  
    &= \exp\left(\pi i \int_{\text{bulk}} (a_{\rd}a_{\rd}+ a_{\rd}a_{\bl}){\hat v}+{\hat v}(a_{\bl}a_{\bl} +a_{\rd}a_{\bl}) + \hat vd\hat v (a_{\bl}+d\hat{v}) +  d\hat{v}\hat vd\hat v  + d\hat{v}\hat{v} a_{\bl}\right) \exp\left(\pi i \int_{\text{bdry}}{\hat{v}}d{\hat v}+\tilde a_{\rd}\cup_1({\hat{v}}d{\hat v}) \right) \\ 
    &= \exp\left(\pi i \int_{\text{bulk}} (a_{\rd}a_{\rd}+ a_{\rd}a_{\bl}){\hat v}+{\hat v}(a_{\bl}a_{\bl} +a_{\rd}a_{\bl}) \right) \exp\left(\pi i \int_{\text{bdry}} \hat{v}a_{\bl}+ a_{\rd}\cup_1({\hat{v}}d{\hat v}) \right) \\ 
     &= \exp\left(\pi i \int_{\text{bulk}} (a_{\rd}a_{\rd}+ a_{\rd}a_{\bl}){\hat v}+{\hat v}(a_{\bl}a_{\bl} +a_{\rd}a_{\bl}) \right)~.
     \end{split}
     \label{eq:commutation with W with boundary}
\end{align}
By combining \eqref{commutation with V boundary}, \eqref{eq:commutation with W with boundary} we get
\begin{align}
    WV S_{X;v}^{\rd}S_{X;v}^{\bl}V^\dagger W^\dagger = S_{X;v}^{\rd}S_{X;v}^{\bl}~,
\end{align}
therefore $U=WV$ is a logical gate.

Now let us study the logical action on the code space. In the logical Pauli $Z$ basis, each basis state $\overline{\ket{0}}, \overline{\ket{1}}$ is labeled by a configuration of $\Z_2^3$ gauge fields $\{a_{\rd},a_{\gr},a_{\bl}\}$. Crucially,  the holonomy of $\Z_2$ gauge fields along the boundary takes the following nontrivial values on the state $\overline{\ket{1}}$,
\begin{align}
    \left(\prod_{e\in \text{bdry}_{\rd\bl}} Z^{\rd}_e\right)\overline{\ket{1}} = \exp\left(\pi i \int_{\text{bdry}_{\rd\bl}} a_{\rd}\right)\overline{\ket{1}} = \exp\left(\pi i \int_{\text{bdry}_{\rd\bl}} a_{\bl}\right)\overline{\ket{1}} = -\overline{\ket{1}}~.
\end{align}
Thus the string of $\text{T}^\dagger$ in the operator $W$ evaluates by a $\text{T}^\dagger$ gate on the states. Meanwhile, one can check that the bulk $\text{CS}$ has a trivial logical action on the state (see Fig.~\ref{fig:triangleboundary}).
Therefore automorphism symmetry $U$ acts on the logical-$Z$ basis as
\begin{equation}
    U\overline{|m\rangle}= e^{-\frac{\pi i m}{4}}\overline{|m\rangle}=\text{T}^\dag\overline{|m\rangle}~,
\end{equation}
with $m=0,1$. This shows that the automorphism symmetry implements the logical $\text{T}^\dagger$ gate.
    
\end{proof}

\begin{figure}[htb]
    \centering
    \includegraphics[width=0.25\linewidth]{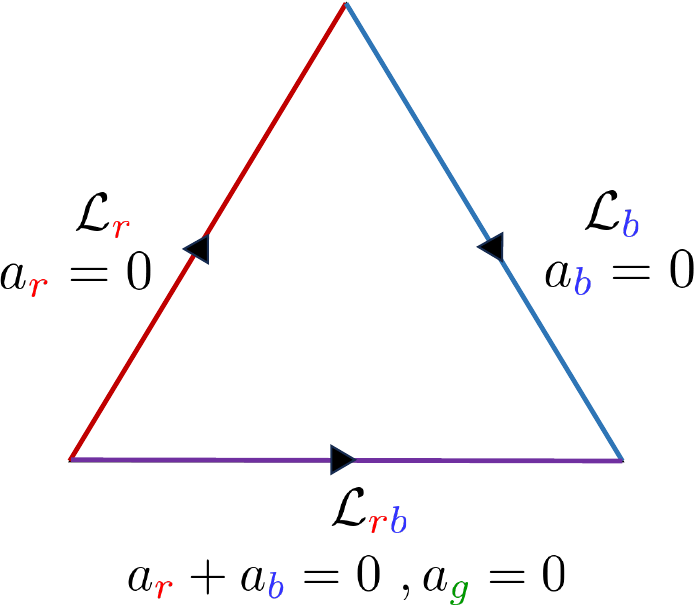}
    \caption{The boundary conditions on gauge fields at the triangle. The logical action of $U=WV$ is evaluated by regarding this triangle as a single 2-simplex, then the bulk cochain $(\tilde a_{\rd}\cup \tilde a_{\bl})/2$ is zero. Meanwhile, the boundary contribution $\tilde a_{\rd}/4$ at
    $\mathcal{L}_{\rd\bl}$ becomes nontrivial, and gives the $\text{T}^\dagger$ action.}
    \label{fig:triangleboundary}
\end{figure}

\section{Details of code switching}

\subsection{\change{Code switching protocol in 2D}}

Here we provide detailed descriptions of the code switching protocol to prepare a T magic state in 2D.
\begin{enumerate}
    \item We start with a logical state $\overline{\ket{+}}$ of the  $\Z_2^{\rd}\times\Z_2^{\bl}$ gauge theory on a triangle with gapped boundaries, which is equivalent to the folded surface code on a triangle and encodes a single qubit.
    These stabilizers are simply obtained by eliminating green qubits from Fig.~\ref{fig:trianglehamiltonian}, and generated by $X$ stabilizers $A_{X;v}^{\rd}, A_{X;v}^{\bl}$ and $Z$ stabilizers $B_{Z;f}^{\rd}, B_{Z;f}^{\bl}$ of folded surface codes. 
    We also initialize the green qubits as $\bigotimes_e \ket{0}^{\gr}_e$. Therefore, the whole stabilizer group is generated by
    \begin{align}
        \left\{A_{X;v}^{\rd}, B_{Z;f}^{\rd},A_{X;v}^{\bl}, B_{Z;f}^{\bl}, Z^{\gr}_e\right\}
        \label{eq:initial stabilizers}
    \end{align}
    
    \item 
    At this point, since $B^{\rd}_{Z;f}=S^{\rd}_{Z;f},B^{\bl}_{Z;f}=S^{\bl}_{Z;f}$ and $Z^{\gr}_e=1$ on all edges, all $Z$ stabilizers of $\mathbb{D}_4$ code $\mathcal{S}_{Z}^{\red{r}},\mathcal{S}_{Z}^{\green{g}},\mathcal{S}_{Z}^{\blue{b}}$ are contained in the stabilizer group. Also, since
    \begin{align}
        S^{\rd}_{X;v} = A^{\rd}_{X;v}\times \prod_{e',e'':\int\tilde v\cup \tilde e'\cup \tilde e''\neq 0}CZ^{\gr,\bl}_{e',e''}~, \quad S^{\bl}_{X;v} = A^{\bl}_{X;v}\times \prod_{e',e'':\int \tilde e'\cup \tilde e''\cup \tilde v\neq 0} CZ^{\rd,\gr}_{e',e''}~,
    \end{align}
    and $CZ^{\gr,\bl}=CZ^{\rd,\bl}=1$ when $Z^{\gr}_e=1$, the $X$ stabilizers $\mathcal{S}_{X}^{\red{r}},\mathcal{S}_{X}^{\blue{b}}$ are also contained in the stabilizer group. The generators of the stabilizer group \eqref{eq:initial stabilizers} are hence rewritten as
     \begin{align}
        \left\{S_{X;v}^{\rd}, S_{Z;f}^{\rd},S_{X;v}^{\bl}, S_{Z;f}^{\bl}, Z^{\gr}_e, S_{Z,f}^{\gr}\right\}~.
    \end{align}    
    We then perform $O(d)$ rounds of syndrome measurements ($d$ is the code distance) of the Clifford stabilizers $S_{X;v}^{\gr}$.
    Since $S^{\gr}_{X;v}$ commutes with the other $X$ stabilizers  $\mathcal{S}_{X}^{\red{r}},\mathcal{S}_{X}^{\blue{b}}$ within the $Z$ stabilizer subspace $\mathcal{S}_{Z}^{\red{r}}=\mathcal{S}_{Z}^{\green{g}}=\mathcal{S}_{Z}^{\blue{b}}=1$, $\mathcal{S}_{X}^{\red{r}},\mathcal{S}_{X}^{\blue{b}}$
    remain in the stabilizer group after this $S^{\gr}_{X;v}$ measurement.
    Therefore, the stabilizer group after the $S^{\gr}_{X;v}$ measurement
    is given by
     \begin{align}
        \left\{S_{X;v}^{\rd}, S_{Z;f}^{\rd},S_{X;v}^{\bl}, S_{Z;f}^{\bl}, (-1)^{m_v}S^{\gr}_{X;v}, S_{Z,f}^{\gr}\right\}~,
    \end{align}  
    where $(-1)^{m_v}=\pm 1$ is the random measurement outcome. This is the stabilizer group of the $\mathbb{D}_4$ Clifford stabilizer code.
    At the same time, we also perform $O(d)$ rounds of measurement for the other Clifford and Pauli $Z$ stabilizers in the entire code; this prepares a state of the $\mathbb{D}_4$ Clifford stabilizer code \cite{Tantivasadakarn2024LRE, verresen2022efficient, Tantivasadakarn2023shortest}.
    
    Meanwhile, we apply the ``just-in-time" decoder from Ref.~\cite{Davydova:2025ylx}.  The non-abelian fluxes corresponding to the detected $Z$-stabilizer syndromes (with eigenvalue $(-1)$) are paired up  in real time  to form closed flux worldline whenever the observation time of the syndromes is comparable to the separation of the syndromes. After the $O(d)$ rounds of syndrome measurements, we use the matching or RG decoder \cite{ducloscianci2010renormalizationgroupdecodingalgorithm} to clean up the remaining Abelian charge syndromes and the corresponding $Z$ errors.  The above procedure is also called a gauging procedure, where one effectively switches the code space into the $\mathbb{D}_4$ code.

    We note that the logical Pauli $\overline{Z}$ operator of the initial folded surface code for $\Z_2^{\rd}\times \Z_2^{\bl}$ commutes with the above $S_{X;v}^{\gr}$ measurement, and the same operator defines the logical $\overline{Z}$ operator of the Clifford stabilizer code after switching. See Fig.~\ref{fig:logicalZ} for an illustration of the logical $\overline{Z}$ operator of the Clifford stabilizer code.
    Therefore the logical state $\ket{\overline{0}}, \ket{\overline{1}}$ of the folded surface code are simply mapped into $\ket{\overline{0}}, \ket{\overline{1}}$ in the Clifford stabilizer code under the switching; the measurement acts by identity on the code space.
    
    \item We act the transversal logical T$^\dagger$ gate $U$ of the Clifford stabilizer code.
    \item We then measure the Pauli $Z^{\gr}_e$ operators on each edges that leads to stabilizers $\{(-1)^{m'_e}Z^{\gr}_e\}$ with random measurement outcome $(-1)^{m'_e}=\pm 1$ on each edge. The edges with $(-1)^{m'_e}=-1$ outcome form a homologically trivial closed loop, therefore can be corrected by acting $S_X^{\gr}$. By acting $S_X^{\gr}$ to flip the strings of $\ket{1}^{\gr}_e$, then the stabilizer becomes 
    \begin{align}
        \left\{A_{X;v}^{\rd}, B_{Z;f}^{\rd},A_{X;v}^{\bl}, B_{Z;f}^{\bl}, Z_{e}^{\gr}\right\}~.
    \end{align} 
    This switches the code back to the $\Z_2^{\rd}\times\Z_2^{\bl}$ surface code followed by $O(d)$ rounds of error corrections.
    The final state is given by $\overline{\text{T}}^\dagger\overline{\ket{+}}$.
\end{enumerate}

\subsection{\change{Code switching from gapped domain walls}}

The above code switching protocol is understood as a composition of gapped domain walls separating the twisted $\Z_2^3$ gauge theory (Clifford stabilizer code) and the untwisted $\Z_2\times\Z_2$ gauge theory (folded surface code). This is illustrated in Fig.~\ref{fig:switching} of the main text.
Here, we provide the perspective of the code and its switching through the anyon condensation that describes gapped boundaries and domain walls.

First, let us describe the gapped boundaries of the Clifford stabilizer code. The gapped boundaries of topological order are characterized essentially by a set of topological excitations that can condense on the boundary, called ``Lagrangian algebra''  \cite{Kapustin:2010hk,davydov2011wittgroupnondegeneratebraided,Levin_2013,Barkeshli_2013}.

The Lagrangian algebras for the three gapped boundaries of the Clifford stabilizer code are described as follows:
\begin{align}
\begin{split}
\label{eq:Lagrangian}
    \mathcal{L}_{\rd} &= 1\oplus e_{\rd}\oplus m_{\bl} \oplus m_{\gr} \oplus m_{\gr\bl}~,\\
    \mathcal{L}_{\bl} &= 1\oplus e_{\bl}\oplus m_{\gr} \oplus m_{\rd} \oplus m_{\rd\gr}~,\\
\mathcal{L}_{\rd\bl} &= 1\oplus e_{\gr}\oplus e_{\rd\bl} \oplus e_{\rd\gr\bl} \oplus 2m_{\rd\bl}~.\\
    \end{split}
\end{align}
The components of each Lagrangian algebra indicates the excitations that are condensed on each boundary. For example, ${\cal L}_{\rd}$ corresponds to condensing $e_{\rd},m_{\bl},m_{\gr},m_{\gr\bl}$ on the boundary.

The gauging measurement switching between the Clifford stabilizer code and folded surface code is understood as a gapped domain wall (represented as (DW) in the main text Fig.~\ref{fig:switching}) which condenses the electric particle $e_{\gr}$ of the twisted $\Z_2^3$ gauge theory.

Then, the transversal T$^\dagger$ gate $U$ of the Clifford stabilizer code corresponds to the anyon permutation automorphism introduced in \eqref{eq:anyonpermutationapp},
\begin{align}
    &m_{\rd}\leftrightarrow m_{\rd \gr},\; m_{\gr}\leftrightarrow m_{\gr}, \; m_{\bl}\leftrightarrow m_{\gr\bl}, \; m_{\rd\bl}\leftrightarrow m_{\rd\bl}~,\cr
    &e_{\rd}\leftrightarrow e_{\rd},\quad e_{\gr}\leftrightarrow e_{\rd\gr\bl},\quad  e_{\bl}\leftrightarrow e_{\bl}~.
\end{align}
The code switching on a folded surface code acts by a composite operator
\begin{align}
    (\text{DW})\times U \times (\text{DW})^\dagger~,
\end{align}
which generates the logical $\text{T}^\dagger$ gate.

As described in the main text, the T gate in the folded surface code through code switching is also constructed in Ref.~\cite{Davydova:2025ylx}. Here we compare two constructions using the anyon condensations and automorphisms of anyon theories.

In their construction, the gapped boundaries of the Clifford stabilizer code on a triangle is given by the following Lagrangian algebras
\begin{align}
\begin{split}
\label{eq:Lagrangian tilde}
    \tilde{\mathcal{L}}_{\rd} &= 1\oplus e_{\rd}\oplus m_{\bl} \oplus m_{\gr} \oplus m_{\gr\bl}~,\\
    \tilde{\mathcal{L}}_{\bl} &= 1\oplus e_{\rd\gr\bl}\oplus m_{\gr\bl} \oplus m_{\rd\bl} \oplus m_{\rd\gr}~,\\
\tilde{\mathcal{L}}_{\rd\bl} &= 1\oplus e_{\gr}\oplus e_{\gr\bl} \oplus e_{\bl} \oplus 2m_{\rd}~.\\
    \end{split}
\end{align}
These boundaries correspond to the side boundaries in (63) of \cite{Davydova:2025ylx} denoted as $\langle \gr,\bl\rangle,\langle \rd \gr,\rd \bl\rangle,\langle \rd\rangle$, respectively.

In Ref.~\cite{Davydova:2025ylx} the logical T gate is then obtained by a composite of distinct gapped domain walls (DW) and (DW') separating $\Z_2\times\Z_2$ and $\mathbb{D}_4$ gauge theory,
\begin{align}
    (\text{DW}') \times (\text{DW})^\dagger~,
\end{align}
where (DW) again corresponds to condensing the electric particle $e_{\gr}$. Meanwhile, the other domain wall (DW') is obtained by fusing the invertible symmetry operator $\tilde{U}$ with (DW), $\text{DW'} = \text{DW}\times \tilde U$, where $\tilde U$ transforms anyons as
\begin{align}
    &m_{\gr}\leftrightarrow m_{\bl},\; m_{\rd}\leftrightarrow m_{\rd}~, \; e_{\gr}\leftrightarrow e_{\bl},\; e_{\rd}\leftrightarrow e_{\rd}~.
    \label{eq:anyonpermutation tilde}
\end{align}

One can then see that the above automorphism $\tilde U$, together with the gapped boundaries $\tilde{\mathcal{L}}_{\rd},\tilde{\mathcal{L}}_{\bl},\tilde{\mathcal{L}}_{\rd\bl}$ are transformed into the ones without tilde by an automorphism permuting anyons as
    \begin{align}
    &m_{\rd}\leftrightarrow m_{\rd \bl},\; m_{\gr}\leftrightarrow m_{\gr\bl}, \; m_{\bl}\leftrightarrow m_{\bl}, \; m_{\rd\gr}\leftrightarrow m_{\rd\gr}~,\cr
    &e_{\rd}\leftrightarrow e_{\rd},\quad e_{\gr}\leftrightarrow e_{\gr},\quad  e_{\bl}\leftrightarrow e_{\rd\gr\bl}~.
\end{align}

Therefore, by rewriting the domain walls $\text{DW'}\times (\text{DW})^\dagger$ as  $\text{DW}\times \tilde U\times  (\text{DW})^\dagger$, the action of $\tilde U$ on the $\mathbb{D}_4$ gauge theory is identified as a logical T gate $U$ by an automorphism \eqref{eq:anyonpermutationapp}.

\section{Logical CS gate in 2D}\label{sec:CS}

The model for logical CS gate has 6 copies of the Clifford stabilizer model. The gauge groups of the first two copies are
denoted by $\Z_2^{\rd}\times \Z_2^{\gr}\times \Z_2^{\bl}$ and $\Z_2'^{\rd}\times \Z_2'^{\gr}\times \Z_2'^{\bl}$, while the gauge groups of the rest of copies are
$(\Z_2^{\rd})^{(i)}\times (\Z_2^{\gr})^{(i)}\times (\Z_2^{\bl})^{(i)}$ for $i=1,2,3,4$.
The region is still a triangle region, with the following boundary conditions as described by breaking the total gauge group to various subgroups:
\begin{itemize}
    \item[(1)] The left boundary ${\cal L}_{(1)}$ breaks it to $\mathbb{Z}_2^{\gr}\times \mathbb{Z}_2^{\bl}\times\mathbb{Z}_2'^{\gr}\times\mathbb{Z}_2'^{\bl}\times \prod_i \left((\mathbb{Z}_2^{\gr})^{(i)}\times (\mathbb{Z}_2^{\bl})^{(i)}\right)$.
    
    \item[(2)] The right boundary ${\cal L}_{(2)}$ breaks it to $\mathbb{Z}_2^{\rd}\times \mathbb{Z}_2^{\gr}\times\mathbb{Z}_2'^{\rd}\times\mathbb{Z}_2'^{\gr}\times \prod_i \left((\mathbb{Z}_2^{\rd})^{(i)}\times (\mathbb{Z}_2^{\gr})^{(i)}\right)$, and
    
    \item[(3)] The bottom boundary ${\cal L}_{(3)}$ breaks it to
\begin{align}
 &\text{diag}(\mathbb{Z}_2^{\rd},\mathbb{Z}_2'^{\bl})\times 
\text{diag}(\mathbb{Z}_2^{\bl},\mathbb{Z}_2'^{\rd})\times\text{diag}(\mathbb{Z}_2^{\gr},\mathbb{Z}_2'^{\gr})\cr
&\times\text{diag}((\mathbb{Z}_2^{\rd})^{(1)},(\mathbb{Z}_2^{\bl})^{(1)},\mathbb{Z}_2^{\rd})
\times
\text{diag}((\mathbb{Z}_2^{\gr})^{(1)},\mathbb{Z}_2^{\bl})\cr 
&\times   \text{diag}((\mathbb{Z}_2^{\rd})^{(2)},(\mathbb{Z}_2^{\gr})^{(2)},\mathbb{Z}_2^{\rd})\times
 \text{diag}((\mathbb{Z}_2^{\bl})^{(1)},\mathbb{Z}_2^{\bl})\cr
 &\times   \text{diag}((\mathbb{Z}_2^{\gr})^{(3)},(\mathbb{Z}_2^{\bl})^{(3)},\mathbb{Z}_2^{\rd})\times
 \text{diag}((\mathbb{Z}_2^{\rd})^{(3)},\mathbb{Z}_2^{\bl})\cr
 &\times   \text{diag}((\mathbb{Z}_2^{\rd})^{(4)},(\mathbb{Z}_2^{\bl})^{(4)},\mathbb{Z}_2^{\rd})\times
 \text{diag}((\mathbb{Z}_2^{\gr})^{(4)},\mathbb{Z}_2^{\bl})~.
\end{align}

\end{itemize}
In terms of the gauge fields for the 6 copies of twisted $\mathbb{Z}_2^3$ gauge theories, these boundary conditions correspond to
(1) $a_{\rd}=0,a_{\rd}'=0,a_{\rd}^{(i)}=0$, (2) $a_{\bl}=0$, $a_{\bl}'=0$, $a_{\bl}^{(i)}=0$, and the boundary (3) is
\begin{align}\label{eqn:bcCS}
&a_{\rd}=a_{\bl}',\;a_{\rd}=a_{\bl}',\;a_{\rd}=a_{\bl}',\;
a_{\bl}=a_{\rd}',\; a_{\gr}=a_{\gr}'\cr
&a_{\rd}^{(1)}=a_{\rd},\; a_{\gr}^{(1)}=a_{\bl}, \; a_{\bl}^{(1)}=a_{\rd}\cr
&a_{\rd}^{(2)}=a_{\rd},\; a_{\gr}^{(2)}=a_{\rd},\; a_{\bl}^{(2)}=a_{\bl}\cr
&a_{\rd}^{(3)}=a_{\bl},\; a_{\gr}^{(3)}=a_{\rd},\; a_{\bl}^{(3)}=a_{\rd}\cr
&a_{\rd}^{(4)}=a_{\rd},\; a_{\gr}^{(4)}=a_{\bl},\; a_{\bl}^{(4)}=a_{\rd}~,    
\end{align}

First we will show that is a well-defined boundary. The topological action $\omega=\pi a_{\rd}\cup a_{\gr}\cup a_{\bl}+\pi a_{\rd}'\cup a_{\gr}'\cup a_{\bl}'$ manifestly becomes trivial for the boundaries (1) and (2). For boundary (3), the action restricted to the unbroken subgroup on the boundary becomes
\begin{align}
    &\pi a_{\rd}\cup a_{\gr}\cup a_{\bl}+\pi a_{\bl}\cup a_{\gr}\cup a_{\rd}\cr 
    &\quad +\pi a_{\rd}\cup a_{\bl}\cup a_{\rd}+\pi a_{\rd}\cup a_{\rd}\cup a_{\bl}+\pi a_{\bl}\cup a_{\rd}\cup a_{\rd}+\pi a_{\rd}\cup a_{\bl}\cup a_{\rd}~.
\end{align}
This is trivial in the cohomology for $\mathbb{Z}_2$ cocycles $a_{\rd},a_{\gr},a_{\bl}$, and thus boundary (3) is also a valid gapped boundary.

Next, let us look at the possible non-contractible Pauli $Z$ strings. 
We will focus on the Pauli $Z$ strings associated with the first two copies of the Clifford stabilizer model.
These strings are generated by the strings shown in Fig.~\ref{fig:cslogicalqubit}. There are two independent Pauli $Z$ strings, and their eigenvalues label the two logical qubits.
We can also pull each junction of the Pauli string in Fig.~\ref{fig:cslogicalqubit} to either the lower left or lower right corner: for the left figure, the logical operator becomes the $Z^{\rd}=Z'^{\bl}$ string on the bottom, while for the right figure it becomes the $Z^{\bl}=Z'^{\rd}$ string on the bottom boundary. In other words, the two logical qubits can also be labeled by the holonomy of the $\mathbb{Z}_2$ gauge fields $a_{\rd}=a_{\bl}'$ and $a_{\bl}=a_{\rd}'$ on the bottom boundary.

\begin{figure}[htb]
    \centering
    \includegraphics[width=0.5\linewidth]{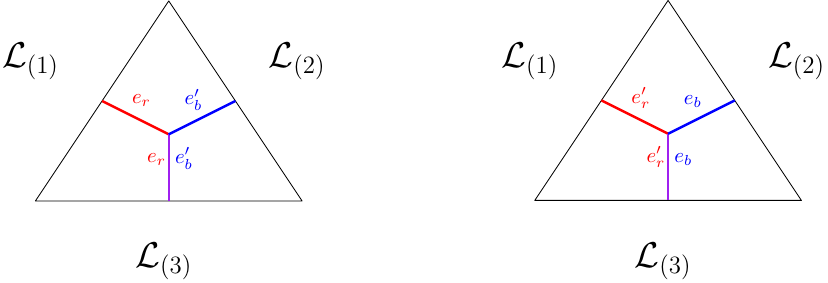}
    \caption{Boundary conditions for CS logical gate from automorphism symmetry. The two logical qubits are labeled by the eigenvalues of two independent Pauli $Z$ operators (left and right), where ${e_{\rd}},{e'_{\rd}}$ strings are $\prod Z^{\rd},\prod Z'^{\rd}$ respectively for the two layers of qubits, and ${e_{\bl}},{e'_{\bl}}$ strings are $\prod Z^{\bl},\prod Z'^{\bl}$ respectively for the two layers of qubits.}
    \label{fig:cslogicalqubit}
\end{figure}

\medskip

Next, let us study the automorphism symmetry in the presence of these boundaries.

\begin{theorem}
    The automorphism symmetry
    \begin{equation}
        \tilde U:=
        U\otimes U'\otimes  
        U(\mathrm{SWAP}_{\rd,\gr})^{(1)}
        \otimes
        U(\mathrm{SWAP}_{\gr,\bl})^{(2)}
        \otimes
        U(\mathrm{SWAP}_{\rd,\gr})^{(3)}
        \otimes
        U(\mathrm{SWAP}_{\gr,\bl})^{(4)}
    \end{equation}
    in the 6-layered Clifford stabilizer model with the boundary condition implements logical CS$^\dagger$ gate. In the above automorphism symmetry, $U,U'$ correspond to the same automorphism in the T gate construction for the first two copies, while the other parts in the automorphism are SWAP symmetries for the other copies.
\end{theorem}
The boundary contribution for this automorphism symmetry is derived as in \cite{Hsin2024:classifying} using the following identity
\begin{equation}\label{eqn:csauto1}
    i^{\int d(\tilde a_{\rd}\cup_1 \tilde a_{\bl})}=
    i^{\int_{\mathrm{bdy (3)}}\tilde a_{\rd}\cup_1 \tilde a_{\bl}}
    (-1)^{\int a_{\rd}^2 \cup_1 a_{\bl} + a_{\bl}^2 \cup_1 a_{\rd}}~,
\end{equation}
where the last sign term in the bulk can be rewritten using the Hirsch identity as
\begin{equation}
    (-1)^{
(a_{\rd} \cup_1 a_{\bl})\cup a_{\rd} + a_{\rd} \cup (a_{\rd} \cup_1 a_{\bl}) 
+ (a_{\bl} \cup_1 a_{\rd}) \cup a_{\rd}+ a_{\rd} \cup (a_{\bl} \cup_1 a_{\rd})
}~.
\end{equation}
The part of the automorphism symmetry $\tilde U$ on the first two copies $U\otimes U'$ gives the gauged SPT defect (\ref{eqn:csauto1}) in addition to the automorphism action. Together with the SWAP automorphisms for the other copies in $\tilde U$, the gauged SPT defects from the automorphism symmetry becomes trivialized by the boundary condition. To see this, we note that the gauged SPT for the SWAP automorphisms can be derived using the following identities: \cite{Hsin:2024pdi}
\begin{align}
    &\pi a_{\gr}^{(i)}\cup a_{\rd}^{(i)}\cup a_{\bl}^{(i)}=
    \pi a_{\rd}^{(i)}\cup a_{\gr}^{(i)}\cup a_{\bl}^{(i)}+\pi d(a_{\rd}^{(i)}\cup_1 a_{\gr}^{(i)})\cup a_{\bl}^{(i)})~,\cr 
    &\pi a_{\rd}^{(i)}\cup a_{\bl}^{(i)}\cup a_{\gr}^{(i)}=
    \pi a_{\rd}^{(i)}\cup a_{\gr}^{(i)}\cup a_{\bl}^{(i)}+\pi a_{\rd}^{(i)}\cup d(a_{\gr}^{(i)}\cup_1 a_{\bl}^{(i)})~.
\end{align}
Thus $\tilde U$ on the other 4 copies has the gauged SPT defect
\begin{equation}
    (-1)^{(a_{\rd}^{(1)}\cup_1 a_{\gr}^{(1)})\cup a_{\bl}^{(1)}}
    (-1)^{a_{\rd}^{(2)}\cup (a_{\gr}^{(2)}\cup_1 a_{\bl}^{(2)})}
    (-1)^{(a_{\rd}^{(3)}\cup_1 a_{\gr}^{(3)})\cup a_{\bl}^{(3)}}
    (-1)^{a_{\rd}^{(4)}\cup (a_{\gr}^{(4)}\cup_1 a_{\bl}^{(4)})}~.
\end{equation}
With the boundary conditions (\ref{eqn:bcCS}) they cancel the sign factors in (\ref{eqn:csauto1}).
Thus the automorphism symmetry $\tilde U$ has the boundary contribution 
\begin{align}
i^{\int_{\mathrm{bdy (3)}}\tilde a_{\rd}\cup_1 \tilde a_{\bl}}~.
\end{align}

For holonomy of $a_{\rd}=a_{\bl}'$ and $a_{\bl}=a_{\rd}'$ on the boundary (3) given by $n_{\rd},n_{\bl}=0,1$, this gives $i^{-n_{\rd}n_{\bl}}$. Meanwhile, one can check that the bulk CS gates evaluate trivially on the nontrivial $\Z_2$ gauge fields.
Thus the finite depth circuit for the automorphism symmetry generates the logical CS$^\dagger$ gate.

\section{$\sqrt{\text{T}}$ gate and Emergent automorphism symmetry in 3D}
\label{app:sqT}

Here we present a non-Clifford stabilizer model for the (3+1)D $\Z_2^4$ gauge theory. The gauge theory consists of four $\Z_2$ gauge fields $a_{\rd},a_{\gr},a_{\bl},a_{\yl}$, with the Dijkgraaf-Witten twist $(-1)^{\int a_{\rd}\cup a_{\yl}\cup a_{\bl}\cup a_{\gr}}$.
The corresponding stabilizer code is supported at a triangulated 3d manifold, with four qubits with colors $\rd,\gr,\bl,\yl$ on each edge. Their Pauli $Z$ operators are related to the $\Z_2$ gauge fields by e.g., $Z_{\rd} = (-1)^{a_{\rd}}$.
The stabilizer group is represented by
\begin{align}
    \mathcal{S} = \left\{ \mathcal{S}_{X}^{\red{r}},\mathcal{S}_{X}^{\green{g}},\mathcal{S}_{X}^{\blue{b}},\mathcal{S}_{X}^{\yl}, 
    \mathcal{S}_{Z}^{\red{r}},\mathcal{S}_{Z}^{\green{g}},\mathcal{S}_{Z}^{\blue{b}},\mathcal{S}_{Z}^{\yl}\right\}~,
\end{align}
with each stabilizer corresponds to 
\begin{align}
    \mathcal{S}_{X}^{\red{r}} = \left\{S_{X;v}^{\red{r}}\right\}~, \quad S_{X;v}^{\red{r}} = \left(\prod_{\partial e\supset v} X^{\rd}_e\right) \prod_{e',e'',e''':\int\tilde v\cup \tilde e'\cup \tilde e''\cup \tilde e'''\neq 0}CCZ^{\yl,\bl,\gr}_{e',e'', e'''}~,
\end{align}
\begin{align}
    \mathcal{S}_{X}^{\gr} = \left\{S_{X;v}^{\gr}\right\}~, \quad S_{X;v}^{\gr} = \left(\prod_{\partial e\supset v} X^{\gr}_e\right) \prod_{e',e'',e''': \int\tilde e'\cup \tilde e''\cup \tilde e'''\cup\tilde v\neq 0}
   CCZ^{\rd,\yl,\bl}_{e',e'',e'''}~,
\end{align}
\begin{align}
    \mathcal{S}_{X}^{\bl} = \left\{S_{X;v}^{\bl}\right\}~, \quad S_{X;v}^{\bl} = \left(\prod_{\partial e\supset v} X^{\bl}_e\right)\prod_{e',e'',e''':\int \tilde e'\cup \tilde e''\cup \tilde v\cup \tilde e'''\neq 0} CCZ^{\rd,\yl,\gr}_{e',e'',e'''}~,
\end{align}
\begin{align}
    \mathcal{S}_{X}^{\yl} = \left\{S_{X;v}^{\bl}\right\}~, \quad S_{X;v}^{\yl} = \left(\prod_{\partial e\supset v} X^{\bl}_e\right)\prod_{e',e'',e''':\int \tilde e'\cup \tilde v\cup \tilde e''\cup \tilde e'''\neq 0} CCZ^{\rd,\bl,\gr}_{e',e'',e'''}~,
\end{align}

\begin{align}
\mathcal{S}_{Z}^{\red{r}} = \left\{S_{Z;f}^{\red{r}}\right\}~,\quad 
    S^{\red{r}}_{Z;f} = \prod_{e\in \partial f} Z_e^{\red{r}}~,
\end{align}
\begin{align}
\mathcal{S}_{Z}^{\gr} = \left\{S_{Z;f}^{\gr}\right\}~,\quad 
    S^{\gr}_{Z;f} = \prod_{e\in \partial f} Z_e^{\gr}~,
\end{align}
\begin{align}
\mathcal{S}_{Z}^{\bl} = \left\{S_{Z;f}^{\bl}\right\}~,\quad 
    S^{\bl}_{Z;f} = \prod_{e\in \partial f} Z_e^{\bl}~.
\end{align}
\begin{align}
\mathcal{S}_{Z}^{\yl} = \left\{S_{Z;f}^{\yl}\right\}~,\quad 
    S^{\yl}_{Z;f} = \prod_{e\in \partial f} Z_e^{\yl}~.
\end{align}

Let us consider this 3D code on a tetrahedron with four gapped boundary conditions shown in Fig.~\ref{fig:tetrahedron}. 
Each gapped boundary is characterized by the boundary conditions on the gauge fields,
\begin{align}
    \begin{split}
        \text{1st boundary} & \quad a_{\bl}= 0~,\\
        \text{2nd boundary} & \quad a_{\rd}= 0~,\\
        \text{3rd boundary} & \quad a_{\gr}= 0~,\\
        \text{4th boundary} & \quad a_{\rd} + a_{\gr} + a_{\bl} = a_{\yl}= 0~,\\
    \end{split}
\end{align}
which are realized by the $Z$-stabilizers at the boundary given by
\begin{align}
    \begin{split}
        \text{$Z$-stabilizer at 1st boundary} & \quad Z_e^{\bl}~,\\
        \text{$Z$-stabilizer at 2nd boundary} & \quad Z_e^{\rd}~,\\
        \text{$Z$-stabilizer at 3rd boundary} & \quad Z_e^{\gr}~,\\
        \text{$Z$-stabilizer at 4th boundary} & \quad Z_e^{\rd}Z_e^{\gr}Z_e^{\bl}~,~ Z^{\yl}~.\\
    \end{split}
\end{align}
The $X$ stabilizers at the boundary are obtained by the ones generated by $\left\{ \mathcal{S}_{X}^{\red{r}},\mathcal{S}_{X}^{\green{g}},\mathcal{S}_{X}^{\blue{b}},\mathcal{S}_{X}^{\yl}\right\}$ that commute with the above boundary $Z$ stabilizers, and truncating such commuting $X$ stabilizers at the boundary. That is, the boundary $X$ stabilizers are given by truncations of the following operators:
\begin{align}
    \begin{split}
        \text{$X$-stabilizer at 1st boundary} & \quad S_{X;v}^{\rd}~,S_{X;v}^{\gr}~,~S_{X;v}^{\yl}~,\\
        \text{$X$-stabilizer at 2nd boundary} & \quad S_{X;v}^{\gr}~,S_{X;v}^{\bl}~,~S_{X;v}^{\yl}~,\\
        \text{$X$-stabilizer at 3rd boundary} & \quad S_{X;v}^{\rd}~,S_{X;v}^{\bl}~,~S_{X;v}^{\yl}~,\\
        \text{$X$-stabilizer at 4th boundary} & \quad S_{X;v}^{\rd}S_{X;v}^{\gr}~,~S_{X;v}^{\gr}S_{X;v}^{\bl}~.\\
    \end{split}
\end{align}
\change{Then the code stores a single logical qubit. To see this, we note that each logical state corresponds to a configuration of $\Z_2^4$ gauge fields up to gauge transformations. With the boundary conditions described above, there is a single nontrivial configuration of $\Z_2^4$ gauge fields on a tetrahedron. This gauge field configuration is characterized by the holonomy of gauge fields $\int a$ along the edges of a tetrahedron, see Fig.~\ref{fig:tetrahedron} for an illustration. This shows that the code space is equivalent to a single logical qubit; $\ket{\overline{0}}$ for trivial gauge fields, and $\ket{\overline{1}}$ for the nontrivial one.}

\change{Analogous to the 2D case, the 3D code has a logical $Z$ operator represented by a junction of Pauli $Z^{\rd}, Z^{\gr}, Z^{\bl}$ string operators meeting at a single point inside a tetrahedron. See Fig.~\ref{fig:3Dlogicals} for an illustration.
This logical $Z$ operator indeed has a $(-1)$ eigenvalue on the logical state $\ket{\overline{1}}$, and hence generates the logical $Z$ operator.
Another way to see this $(-1)$ eigenvalue is to notice another topological operator given by a network of magnetic surface operators, as shown in Fig.~\ref{fig:3Dlogicals}. This operator anti-commutes with the above string operator due to the Aharanov-Bohm phase, therefore the string operator generates the logical Pauli $Z$ gate.}

The code then hosts a logical $\sqrt{\text{T}}=\text{diag}(1, e^{2\pi i/16})$. This logical gate again has an expression $U=WV$,
with
\begin{align}
   V= \bigotimes_e \text{CNOT}^{(\rd,\yl)}_{e} \text{CNOT}^{(\gr,\yl)}_{e}\text{CNOT}^{(\bl,\yl)}_{e}~,
\end{align}

\begin{align}
    W = \exp\left(\pi i \int_{\text{bulk}} \frac{\tilde a_{\rd}\cup \tilde a_{\bl}\cup \tilde a_{\gr} }{2}+a_{\rd}\cup(a_{\gr}\cup_1 a_{\bl})\cup a_{\gr} - \pi i\int_{\text{bdry}_{4}} \frac{\tilde a_{\rd}\cup \tilde a_{\gr}}{4} + \pi i\int_{\text{hinge}_{1,4}}\frac{\tilde a_{\rd}}{8}\right)~,
    \label{expression of W in 3D}
\end{align}
which are expressed by CCS gates in the 3D bulk, CT$^\dagger$ gates on the 4th boundary, and $\sqrt{\text{T}}$ gates on the 1D hinge between the 1st, 4th boundary.

The transversal CNOT operator $V$ permutes the Pauli operators and gauge fields as
\begin{align}
\begin{split}
    & X_e^{\rd}\leftrightarrow X_e^{\rd}X_e^{\yl}, \quad X_e^{\gr}\leftrightarrow X_e^{\gr}X_e^{\yl}, \quad X_e^{\bl}\leftrightarrow X_e^{\bl}X_e^{\yl}~,\quad X_e^{\yl}\leftrightarrow X_e^{\yl}~,\\
    & Z_e^{\rd} \leftrightarrow Z_e^{\rd}, \quad Z_e^{\gr} \leftrightarrow Z_e^{\gr}~, \quad Z_e^{\bl} \leftrightarrow Z_e^{\bl}~,\quad Z_e^{\yl} \leftrightarrow Z_e^{\rd}Z_e^{\gr}Z_e^{\bl}Z_e^{\yl}~.
\end{split}
\end{align}
The form of the operator $W$ is obtained by the following procedure:
\begin{enumerate}
    \item First, the operator $V$ shifts the Dijkgraaf-Witten twist by
    \begin{align}
    \begin{split}
& a_{\rd}\cup (a_{\rd}+a_{\gr}+a_{\bl}+a_{\yl})\cup a_{\bl}\cup a_{\gr} - a_{\rd}\cup a_{\yl}\cup a_{\bl}\cup a_{\gr} \\
& = d\left(\frac{\tilde a_{\rd}\cup \tilde a_{\bl}\cup \tilde a_{\gr} }{2}+a_{\rd}\cup(a_{\gr}\cup_1 a_{\bl})\cup a_{\gr}\right)~.
\end{split}
    \end{align}
Therefore, according to the discussions in Sec.~\ref{app:automorphism path integral}, the bulk integral in $W$ gives an expression of the logical operator in the bulk.

\item The boundary contributions are obtained by trivializing the bulk cochain at the boundary, following the prescription in Sec.~\ref{app:hinge}. The bulk 3-cochain
\begin{align}
    \frac{\tilde a_{\rd}\cup \tilde a_{\bl}\cup \tilde a_{\gr} }{2}+a_{\rd}\cup(a_{\gr}\cup_1 a_{\bl})\cup a_{\gr}
\end{align}
vanishes at boundaries except for the 4th boundary, where the 3-cochain trivializes as $\frac{1}{4}d(\tilde a_{\rd}\cup\tilde a_{\gr})$ under the boundary condition $a_{\rd}+a_{\gr}+a_{\bl}=0$ mod 2. This gives the boundary integral in the expression \eqref{expression of W in 3D}.

\item The hinge contributions are obtained by trivializing the boundary cochains at the hinge. At the hinges surrounding the 4th boundary, the boundary 2-cochain $\frac{1}{4}(\tilde a_{\rd}\cup\tilde a_{\gr})$ becomes zero except for the hinge between the 1st and 4th boundary. At this hinge, the gauge fields satisfy the boundary conditions at both 1st and 4th boundaries, hence $a_{\rd}+a_{\gr}=a_{\bl}=0.$ The boundary cochain then trivializes as $\frac{1}{8}d\tilde a_{\rd}$, which gives the hinge contribution in \eqref{expression of W in 3D}.

\end{enumerate}

The logical action of the above operator $U=WV$ is obtained by evaluating the integral in $W$ on a nontrivial $\Z_2$ gauge fields of a tetrahedron shown in Fig.~\ref{fig:tetrahedron}, that corresponds to the logical state $\ket{1}$. By regarding a tetrahedron as a single 3-simplex, one can check that the bulk and boundary cochains become trivial on the tetrahedron. The only nontrivial contribution arises from the hinge, where the operator $\tilde{a}_{\rd}/8$ evaluates as $1/8$. This implies that the operator $U=WV$ generates the $\sqrt{\text{T}}$ gate.

\begin{figure}[htb]
    \centering
    \includegraphics[width=0.65\linewidth]{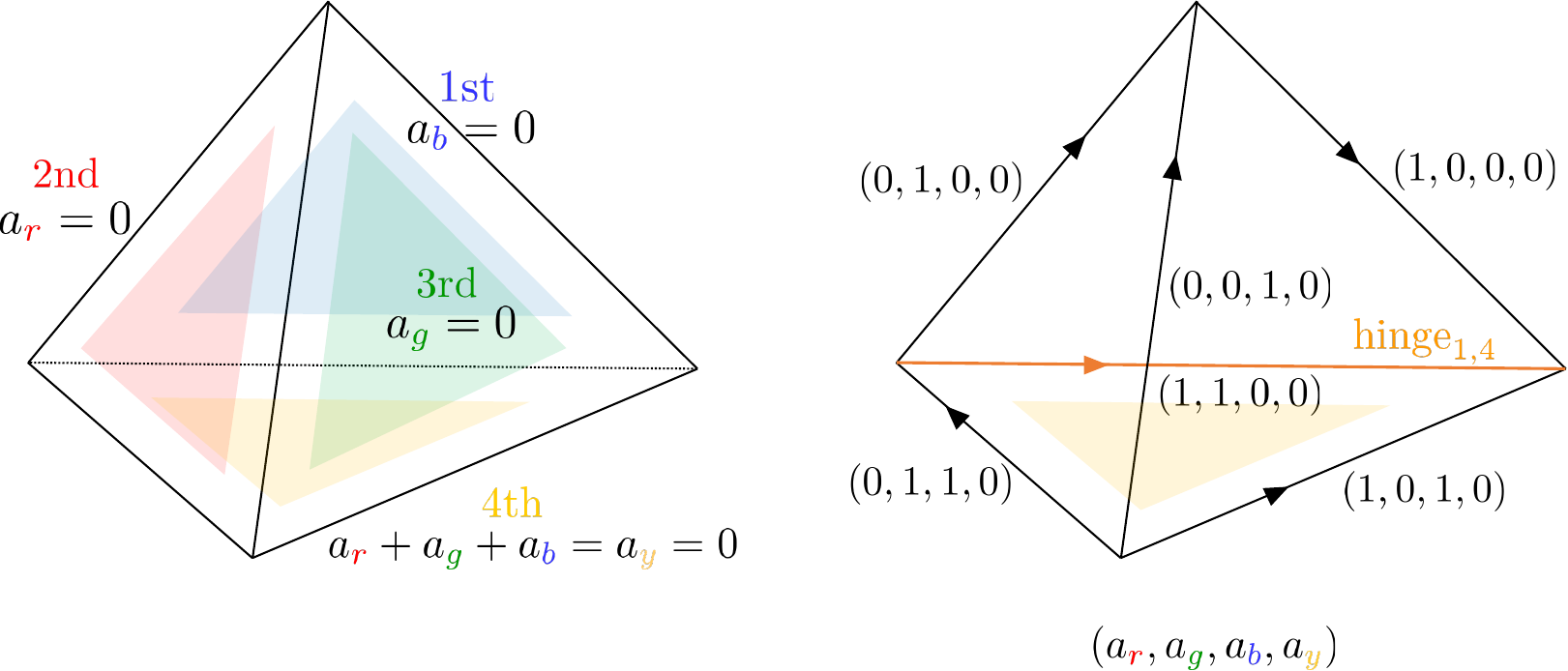}
    \caption{Left: a tetrahedron is bounded by four gapped boundary conditions on each face. Each gapped boundary is characterized by the boundary conditions of $\Z_2$ gauge fields.  Right: The holonomy of the nontrivial $\Z_2$ gauge fields $(a_{\rd}, a_{\gr}, a_{\bl}, a_{\yl})$ on each edge. This nontrivial $\Z_2$ gauge field corresponds to the logical state $\ket{\overline{1}}$. By regarding this tetrahedron as a single simplex and evaluating the integral of $W$, one can see that the integral is nontrivial only at the hinge$_{1,4}$ which gives the $\sqrt{\text{T}}$ gate action.}
    \label{fig:tetrahedron}
\end{figure}

\begin{figure}[htb]
    \centering
    \includegraphics[width=0.65\linewidth]{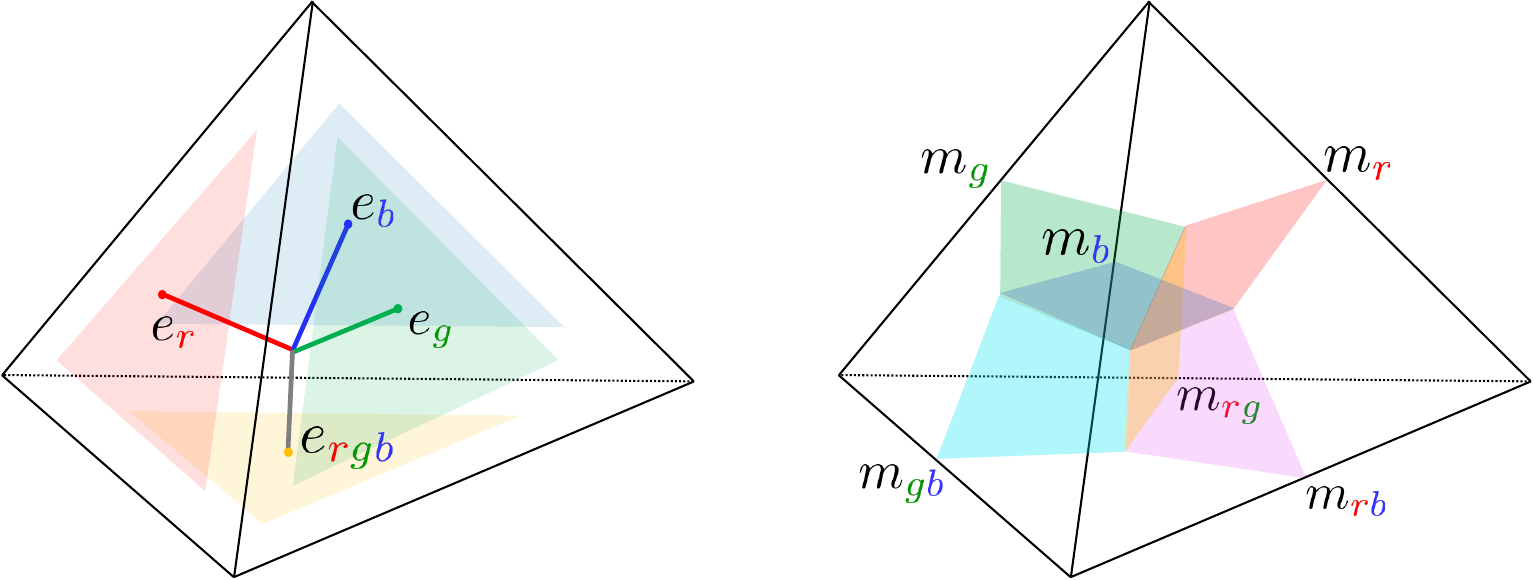}
    \caption{ \change{Left: A logical Pauli $Z$ operator is given by a junction of three strings of electric charges $e_{\rd}, e_{\gr}, e_{\bl}$ stemming from 2nd, 3rd, 1st boundary respectively, fusing inside a tetrahedron into $e_{\rd\gr\bl}$, and terminating at the 4th boundary. Each line operator is a string of $Z^{\rd}, Z^{\gr}, Z^{\bl}, Z^{\rd}Z^{\gr}Z^{\bl}$ along a line. This operator has $(-1)$ eigenvalue on the nontrivial $\Z_2^4$ gauge field depicted in Fig.~\ref{fig:tetrahedron}, hence generates the logical $Z$ gate. Right: There is another non-invertible topological operator formed by a network of magnetic surface operators. It anti-commutes with the logical Pauli $Z$ operator.} }
    \label{fig:3Dlogicals}
\end{figure}

\vfill

\end{document}